\documentclass[11pt]{article}
\usepackage{etex}
\usepackage{fullpage}
\usepackage{extarrows}

\usepackage[numbers,sort&compress]{natbib}
\usepackage{amsfonts,amssymb,mathabx}
\usepackage{times}

\newcommand{\U}{}

\newcommand\appendices{\appendix}

\usepackage{ifthen}
\usepackage{suffix}
\usepackage{xspace}
\usepackage{algorithmicx}
\usepackage{paralist}
\usepackage{enumitem}
\usepackage[noend]{algpseudocode}

\usepackage{hyperref}
\usepackage{float}
\usepackage{utf8math}
\usepackage{ttquot,macros,macros2}
\usepackage{txtt}

\usepackage{bcprules}
\usepackage{stmaryrd}

\usepackage{empheq}

\usepackage[font=small]{caption}
\usepackage{subcaption}
\hypersetup{
    colorlinks,
    urlcolor=black,
    citecolor=black,
    filecolor=black,
    linkcolor=black
}
\usepackage[table,usenames,dvipsnames]{xcolor}
\definecolor{ForestGreen}{rgb}{0, 0.5, 0}
\usepackage{array}
\usepackage{arydshln}
\usepackage{tabu}
\usepackage{multirow}
\usepackage{tikz}

\usepackage{graphicx}

\usepackage{listings}
\usepackage{amsthm}
\theoremstyle{plain}
\newtheorem{theorem}{Theorem}

\newtheorem{lemma}{Lemma}
\newtheorem*{rethm}{\retitle}

\theoremstyle{definition}
\newtheorem{definition}{Definition}

\newcommand{\retheoremType}[2]{%
  \newenvironment{#1}[1]{%
    \def\retitle{#2 \ref{##1}}%
    \begin{rethm}%
  }{\end{rethm}}%
}
\retheoremType{retheorem}{Theorem}
\retheoremType{reproposition}{Proposition}

\usepackage{mathpartir}

\usepackage{refcount,cleveref}
\crefformat{section}{\mbox{Section #1}}
\crefformat{appendix}{\mbox{Appendix #1}}
\crefformat{subsection}{\mbox{Section #1}}
\crefrangeformat{section}{\mbox{Sections #3#1#4--#5#2#6}}
\crefrangeformat{subsection}{\mbox{Sections #3#1#4--#5#2#6}}
\crefrangeformat{subsubsection}{\mbox{Sections #3#1#4--#5#2#6}}
\crefmultiformat{section}{\mbox{Sections #2#1#3}}{ and \mbox{#2#1#3}}
{, \mbox{#2#1#3}}{, and \mbox{#2#1#3}}
\crefmultiformat{subsection}{\mbox{Sections #2#1#3}}{ and \mbox{#2#1#3}}
{, \mbox{#2#1#3}}{, and \mbox{#2#1#3}}
\crefmultiformat{subsubsection}{\mbox{Sections #2#1#3}}{ and \mbox{#2#1#3}}
{, \mbox{#2#1#3}}{, and \mbox{#2#1#3}}

\crefname{table}{Table}{Table}

\crefformat{figure}{Figure~#1}
\crefrangeformat{figure}{\mbox{Figures #3#1#4--#5#2#6}}
\crefmultiformat{figure}{\mbox{Figures #2#1#3}}{ and \mbox{#2#1#3}}
{, \mbox{#2#1#3}}{, and \mbox{#2#1#3}}
\crefformat{equation}{(#1)}
\crefrangeformat{equation}{\mbox{(#1)--(#2)}}

\title{Nonmalleable Information Flow Control: Technical Report}

\author{Ethan Cecchetti \\
        Department of Computer Science \\
        Cornell University \\
        \href{mailto:ethan@cs.cornell.edu}{\large\tt ethan@cs.cornell.edu}
   \and Andrey C. Myers \\
        Department of Computer Science \\
        Cornell University \\
        \href{mailto:andru@cs.cornell.edu}{\large\tt andru@cs.cornell.edu}
   \and Owen Arden \\
        Department of Computer Science \\
        University of California, Santa Cruz\thanks{Work done while at Harvard University} \\
        \href{mailto:owen@soe.ucsc.edu}{\large\tt owen@soe.ucsc.edu}}
\date{}

\begin{document}
  \ifx\case\undefined
    \newenvironment{case}[1][]{\vspace{6pt}\noindent Case #1:}{}
  \fi

  \advance\belowcaptionskip -0.5em
  \advance\abovecaptionskip -1em

  \maketitle

  \begin{abstract}
  Noninterference is a popular semantic security condition because it
offers strong end-to-end guarantees, it is inherently compositional,
and it can be enforced using a simple security type system.
Unfortunately, it is too restrictive for real systems.
Mechanisms for downgrading information are needed to capture real-world
security requirements, but downgrading eliminates the strong compositional
security guarantees of noninterference.

We introduce _nonmalleable information flow_, a new
formal security condition that generalizes noninterference to permit
controlled downgrading of both confidentiality and integrity.
While previous work on robust declassification
prevents adversaries from exploiting the downgrading of confidentiality,
our key insight is _transparent endorsement_, a mechanism
for downgrading integrity while defending against adversarial
exploitation.  Robust declassification appeared to break the duality of
confidentiality and integrity by making confidentiality depend on
integrity, but transparent endorsement makes
integrity depend on confidentiality, restoring this duality.
We show how to extend a
security-typed programming language
with transparent endorsement and
prove that this static type system enforces nonmalleable information
flow, a new security property that subsumes robust declassification and transparent
endorsement.
Finally, we describe an implementation of this type
system in the context of Flame, a flow-limited authorization plugin
for the Glasgow Haskell Compiler.

  \end{abstract}

  \section{Introduction}
\label{sec:intro}

An ongoing foundational challenge for computer security is to discover
rigorous---yet sufficiently flexible---ways to specify what it means for a
computing system to be secure. Such security conditions should be
_extensional_, meaning that they are based on the externally
observable behavior of the system rather than on unobservable details
of its implementation. To allow security enforcement mechanisms to
scale to large systems, a security condition should also be
_compositional_, so that secure subsystems remain secure when
combined into a larger system.

_Noninterference_, along with many variants~\cite{GM82,sm-jsac}, has
been
a popular security condition precisely because it is both
extensional and compositional. 
Noninterference forbids all flows of information from ``high'' to
``low'', or more generally, flows of information that violate a
lattice policy~\cite{denning-lattice}.

Unfortunately, noninterference is also known to be too restrictive for
most real systems, which need fine-grained control
over when and how information flows.
Consequently, most implementations of information flow control
introduce _downgrading_ mechanisms to allow
information to flow contrary to the lattice
policy. Downgrading confidentiality is called
_declassification_, and downgrading integrity---that is, treating
information as more trustworthy than information that has influenced
it---is known as _endorsement_~\cite{zznm02}.

Once downgrading is permitted, noninterference is lost. The natural
question is whether downgrading can nevertheless be constrained
to guarantee that systems still satisfy some meaningful,
extensional, and compositional security conditions. This paper shows
how to constrain the use of
both declassification and endorsement in a
way that ensures such a security condition holds.

Starting with the work of Biba~\cite{integrity}, integrity has
often been viewed as dual to confidentiality. Over time,
that simple duality has eroded. In particular, work on _robust
declassification_~\cite{zm01b,zznm02,msz06,cm06,am11} has shown that
in the presence of declassification, confidentiality depends on
integrity. It is dangerous to give the adversary the ability to influence
declassification, either by affecting the data that is declassified
or by affecting the decision to perform declassification. By
preventing such influence, robust
declassification stops the adversary from _laundering_ confidential
data through existing declassification operations.
Operationally, languages prevent laundering by restricting declassification to
high integrity program points.
Robust
declassification can be enforced using a modular type system and is therefore
compositional.

This paper introduces a new security condition, _transparent
endorsement_, which is dual to robust declassification: it
controls endorsement by using _confidentiality_ to limit the possible
relaxations of _integrity_.
Transparent endorsement
prevents an agent from endorsing information that the provider of the
information could not have seen.
Such endorsement is dangerous because it permits the provider 
to affect flows from the endorser's own secret information into trusted information.
This restriction on endorsement enforces an often-implicit justification for 
endorsing untrusted inputs in high-integrity, confidential computation (e.g., a password checker): 
low-integrity inputs chosen by an attacker should be chosen without knowledge of secret information.

A similar connection between the confidentiality and integrity of
information arises in cryptographic settings.
A _malleable_ encryption scheme is one where a ciphertext encrypting one value
can be transformed into a ciphertext encrypting a related value.
While sometimes malleability is intentional (e.g., _homomorphic_ encryption),
an attacker's ability to generate ciphertexts makes malleable encryption
insufficient to authenticate messages or validate integrity.
Nonmalleable encryption schemes~\cite{ddn03} prevent such attacks.
In this paper, we combine robust declassification and
transparent endorsement into a new security condition, _nonmalleable information flow_,
which prevents analogous attacks in an information flow control setting.

The contributions of this paper are as follows:
\begin{itemize}
  \item We give example programs showing the need for a
        security condition that controls endorsement of
        secret information.

  \item We generalize _robust declassification_ to programs including complex
    data structures with heterogeneously labeled data.

  \item We identify _transparent endorsement_ and _nonmalleable information
    flow_, new extensional security conditions for programs including
    declassification and endorsement.

  \item We present a core language, \nmlang, which provably enforces robust
    declassification, transparent endorsement, and nonmalleable information
    flow.

  \item We present the first formulation of robust declassification as
    a _4-safety hyperproperty_, and define two new 4-safety
    hyperproperties for transparent endorsement and nonmalleable
    information flow, the first time information security conditions have been
    characterized as $k$-safety hyperproperties with $k > 2$.

  \item We describe our implementation of \nmlang using Flame, a flow-limited
    authorization library for Haskell and adapt an example of the Servant web
    application framework, accessible online at \demourl.
\end{itemize}

We organize the paper as follows.
Section~\ref{sec:motivation} provides examples of vulnerabilities in
prior work. Section~\ref{sec:background} reviews relevant background. 
Section~\ref{sec:typesys} introduces our approach for controlling dangerous endorsements,
and Section~\ref{sec:nmlang} presents a syntax, semantics, and type system for \nmlang.
Section~\ref{sec:condition} formalizes our security conditions and
Section~\ref{sec:hyperproperty} restates them as hyperproperties.
Section~\ref{sec:impl} discusses our Haskell implementation, 
Section~\ref{sec:related} compares our approach to related work,
and Section~\ref{sec:conclusions} concludes.

\section{Motivation}
\label{sec:motivation}

To motivate the need for an additional security condition and give
some intuition about transparent endorsement, we give three short
examples. Each example shows code that type-checks under existing information-flow
type systems even though it contains insecure information flows,
which we are able to characterize in a new way.

These examples use the notation of the flow-limited authorization model
(FLAM)~\cite{flam}, which offers an expressive way to state both
information flow restrictions and authorization policies. However,
the problems observed in these examples are not specific to
FLAM; they arise in all previous
information-flow models that support downgrading (e.g.,~\cite{ml-tosem,
liopriv,Paralocks,flume,asbestos,histar,laminar}).
The approach
in this paper can be applied straightforwardly to the decentralized label model
(DLM)~\cite{ml-tosem}, and with more effort, to DIFC models that are
less similar to FLAM.
While some previous models lack a notion of integrity, from our
perspective they are even worse off, because they effectively 
allow _unrestricted_ endorsement.

In FLAM, principals and information flow labels occupy the same space.
Given a principal (or label) $p$, the notation $p^{→}$ denotes the
confidentiality projection of $p$, whereas the notation $p^{←}$
denotes its integrity projection. Intuitively, $p^{→}$ represents the 
authority to decide where $p$'s secrets may flow _to_, whereas
$p^{←}$ represents the authority to decide where information trusted by $p$ may flow _from_.  Robust declassification ensures that
the label $p^{→}$ can be removed via declassification only in code
that is trusted by $p$; that is, with integrity $p^{←}$.

Information flow policies provide a means to specify security
requirements for a program, but not an enforcement mechanism.  For
example, confidentiality policies might be implemented using
encryption and integrity policies using digital signatures.
Alternatively, hardware security mechanisms such as memory protection
might be used to prevent untrusted processes from reading confidential
data.  The following examples illustrate issues that would arise in
many information flow control systems, regardless of the enforcement
mechanism.

\subsection{Fooling a password checker}
\label{sec:foolpwd}

\begin{figure}
\begin{lstlisting}[xleftmargin=1.5em,numbers=left]
String(*$_T$*) password;

boolean(*$_{T^{←}}$*) check_password(String(*$_{T^{→}}$*) guess) {
    boolean(*$_T$*) endorsed_guess = endorse(guess, (*$T$*)); (*\label{li:pwd:endorse}*)
    boolean(*$_T$*) result = (endorsed_guess == password);
    return declassify(result, (*$T^{←}$*)); (*\label{pwd-declassify}*)
}
\end{lstlisting}
\caption{A password checker with malleable information flow}
\label{fig:insecure-pwd}
\end{figure}

Password checkers are frequently used as an example of necessary and justifiable
downgrading.
However, incorrect downgrading can allow an attacker who does not know the
password to authenticate anyway.
Suppose there are two principals, a fully trusted
principal $T$ and an untrusted principal $U$. The following
information flows are then secure: $U^{→} ⊑ T^{→}$ and
$\integ T ⊑ \integ U$. Figure \ref{fig:insecure-pwd} shows in pseudo-code how
we might erroneously implement a password checker in a security-typed language
like Jif~\cite{myers-popl99}. Because this pseudo-code would satisfy
the type system, it might appear to be secure.

The argument "guess" has no integrity because it is
supplied by an untrusted, possibly adversarial source. It is necessary to
declassify the result of the function (at line~\ref{pwd-declassify})
because the result indeed leaks a little
information about the password.
Robust declassification, as enforced in Jif, demands that the untrusted "guess"
be endorsed before it can influence information released by declassification.

Unfortunately, the "check_password" policy does not prevent faulty or malicious
(but well-typed) code from supplying "password" directly as the argument,
thereby allowing an attacker with no knowledge of the correct password to
``authenticate.''
Because "guess" is labeled as secret ($T^{→}$), a flow of information from
"password" to "guess" looks secure to the type system, so this severe
vulnerability could remain undetected.
To fix this we would need to make "guess" less secret, but no prior work has
defined rules that would require this change.
The true insecurity, however, lies on line~\ref{li:pwd:endorse}, which
erroneously treats sensitive information as if the attacker had constructed it.
We can prevent this insecurity by outlawing such endorsements.

\subsection{Cheating in a sealed-bid auction}
\label{sec:cheat-auction}

\begin{figure}
\begin{center}
\begin{tikzpicture} [auto, align=center, text width=4em,node distance=4.75em, thick, x=3.5em]
  \node (1) at (0,0.5) {A};
  \node (11) at (0,0) {"a_bid"};
  \node (12) at (0,-1.5) {};
  \node (13) at (0,-2.5) {"b_bid"};
  \node (13B) at (6,-2.5) {};
  \node (14) at (0,-3.5) {"b_bid"};

  \node (T) at (3,0.5) {T};
  \node (TA) at (3,-0.5) {"a_bid"};
  \node (TAb) at (3,-0.7) {};
  \node (TB) at (3,-2) {"b_bid"};
  \node (TBb) at (3,-2.2) {};
  \node (T4) at (3,-3.5) {Bids};

  \node (2) at (6,0.5) {B};
  \node (21) at (6,-1) {"a_bid"};
  \node (21A) at (0,-1) {};
  \node (22) at (6,-1.5) {"b_bid"};
  \node (24) at (6,-3.5) {"a_bid"};
  \node (T1) at (3,0.5) {T};

 \path[->,every node/.style={font=\sffamily\small}]
    (11) edge node[xshift=-1.25em, midway, above, sloped] {\scriptsize\begin{tabular}{c}$A^{←}∧(A∧B)^{→}$ \end{tabular}} (TA)
    (TAb) edge node[midway, above, sloped] {\scriptsize\begin{tabular}{c}$(A∧B)$ \end{tabular}} (21)
    (TAb) edge node[midway, above, sloped] {\scriptsize\begin{tabular}{c}$(A∧B)$ \end{tabular}} (21A)
    (22) edge node[xshift=-1.25em, midway, above, sloped] {\scriptsize\begin{tabular}{c}$B^{←}∧(A∧B)^{→}$ \end{tabular}} (TB)
    (TBb) edge node[midway, above, sloped] {\scriptsize\begin{tabular}{c}$(A∧B)$ \end{tabular}} (13)
    (TBb) edge node[midway, above, sloped] {\scriptsize\begin{tabular}{c}$(A∧B)$ \end{tabular}} (13B)
    (T4) edge node[xshift=-2.5em, above] {\scriptsize\begin{tabular}{c}open "b_bid"\\ $(A ∧ B)^{←} ∧ (A ∨ B)^{→} $ \end{tabular} } (14)
    (T4) edge node[xshift=-2.5em, above] {\scriptsize\begin{tabular}{c}open "a_bid"\\ $(A ∧ B)^{←} ∧ (A ∨ B)^{→} $ \end{tabular} } (24);
\end{tikzpicture}
\end{center}
\caption{Cheating in a sealed-bid auction. Without knowing Alice's bid, Bob can always win by setting \texttt{b\_bid := a\_bid + 1}}
\label{fig:auction}
\end{figure}
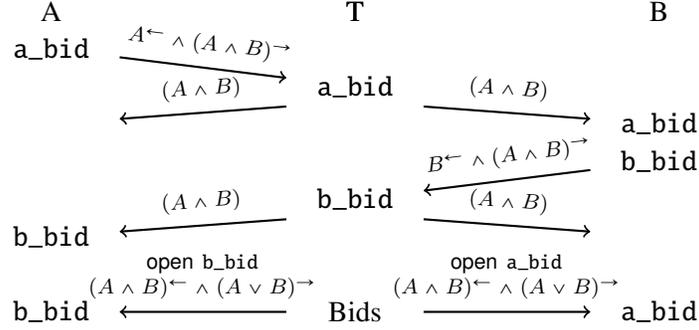
Imagine that two principals $A$ and $B$ (Alice and Bob) are engaging in a
two-party sealed-bid auction administered by an auctioneer $T$ whom they both
trust. Such an auction might be implemented using cryptographic commitments
and may even simulate $T$ without need of an actual third party.
However, we
abstractly specify the information security requirements that such a scheme
would aim to satisfy. Consider
the following sketch of an auction protocol, illustrated in Figure~\ref{fig:auction}:

\begin{enumerate}[leftmargin=*]
  \item \label{auc:li:A-bid}
    $A$ sends her bid "a_bid" to $T$ with label $A^{←}∧(A∧B)^{→}$.
    This label means "a_bid" is trusted only by those who trust $A$
    and can be viewed only if _both_ $A$ and $B$ agree to release it.

  \item \label{auc:li:T-accept}
    $T$ accepts "a_bid" from $A$ and uses his authority to endorse
    the bid to label $(A∧B)^{←}∧(A∧B)^{→}$ (identically, $A∧B$).
    The endorsement prevents any further unilateral modification to the bid by $A$.
    $T$ then broadcasts this endorsed "a_bid" to $A$ and $B$.
    This broadcast corresponds to an assumption that network
    messages can be seen by all parties.

  \item \label{auc:li:B-bid}
    $B$ constructs "b_bid" with label $B^{←}∧(A∧B)^{→}$ and sends it to $T$.

  \item \label{auc:li:A-accept}
    $T$ endorses "b_bid" to $A∧B$ and broadcasts the result.
  \item \label{auc:li:declassify}
    $T$ now uses its authority to declassify both bids and send them
    to all parties. Since both bids have high integrity,
    this declassification is legal according to existing typing rules
    introduced to enforce (qualified) robust
    declassification~\cite{msz06,cm06,flam}.
\end{enumerate}

Unfortunately, this protocol is subject to attacks analogous to _mauling_ in 
malleable cryptographic schemes~\cite{ddn03}:  $B$ can always win the
auction with the minimal winning bid.  In Step~\ref{auc:li:B-bid}
nothing prevents $B$ from constructing "b_bid" by adding 1 to "a_bid",
yielding a new bid with label $B^{←}∧(A∧B)^{→}$ (to modify the value,
$B$ must lower the value's integrity as $A$ did not authorize the modification).

Again an insecurity stems from erroneously endorsing overly secret
information. In step~\ref{auc:li:A-accept}, $T$ should not endorse
"b_bid" since it could be based on confidential information inaccessible
to $B$---in particular, "a_bid".
The problem can be fixed by giving $A$'s bid the label $A^{→}∧A^{←}$
(identically, just $A$), but existing information flow systems impose
no such requirement.

\subsection{Laundering secrets}
\label{sec:launder-sec}

Wittbold and Johnson~\cite{wj90} present an
interesting but insecure program:

\begin{lstlisting}[morekeywords={output},xleftmargin=1.5em,numbers=left]
while (true) do {
    x = 0 (*$[\!]$*) x = 1;  // generate secret probabilistically
    output x to (*$H$*);
    input y from (*$H$*);   // implicit endorsement
    output x (*⊕*) (y mod 2) to (*$L$*)
}
\end{lstlisting}

In this code, there are two external agents, $H$ and $L$. Agent $H$ is
intended to have access to secret information, whereas $L$ is not.
The code generates a secret by assigning to the variable "x" a
nondeterministic, secret value that is either 0 or 1. The choice of
"x" is assumed not to be affected by the
adversary. Its value is used as a one-time pad to conceal the secret
low bit of variable "y".

Wittbold and Johnson observe that this code permits an adversary to launder one bit of
another secret variable "z" by sending "z⊕x" as the value
read into "y".  The low bit of "z" is then the output to $L$.

Let us consider this classic example from the viewpoint of a modern
information-flow type system that enforces robust declassification.
In order for this code to type-check, it must declassify the value
"x⊕(y mod 2)". Since the attack depends on "y" being affected by
adversarial input from $H$, secret input from $H$ must be low-integrity (that is,
its label must be $H^{→}$). But if it is low-integrity, this input
(or the variable "y") must be endorsed to allow the declassification
it influences. As in the previous two examples, the endorsement of
high-confidentiality information enables exploits.

\section{Background}
\label{sec:background}

We explore nonmalleable information flow in the context of a simplified
version of FLAM~\cite{flam}, so we first present some background.
FLAM provides a unified model for reasoning about both information
flow and authorization. Unlike in previous models, principals and
information flow labels in FLAM are drawn from the same set $\L$. The
interpretation of a label as a principal is the least powerful
principal trusted to enforce that label. The interpretation of a
principal as a label is the strongest information security policy that
principal is trusted to enforce. We refer to elements of $\L$
as principals or labels depending on whether we are talking about
authorization or information flow.

Labels (and principals) have both confidentiality and integrity
aspects.  A label (or principal) $\ell$ can be projected to capture just
its confidentiality ($\ell^{→}$) and integrity ($\ell^{←}$) aspects. 

The information flow ordering $⊑$ on labels (and principals) describes
information flows that are secure, in the direction of increasing
confidentiality and decreasing integrity. The orthogonal trust
ordering $≽$ on principals (and labels) corresponds to increasing
trustedness and privilege: toward
increasing confidentiality and _increasing_ integrity.
We read $ℓ⊑ℓ'$ as ``$ℓ$ flows to $ℓ'$'', meaning $ℓ'$ specifies a policy at
least as restrictive as $ℓ$ does.
We read $p≽q$ as ``$p$ acts for $q$'', meaning that $q$ delegates to $p$.

The information flow and the trust orderings each define a lattice over
$\L$, and these lattices lie intuitively at right angles to one another. The least
trusted and least powerful principal is $⊥$, (that is, $p≽⊥$ for all
principals $p$), and the most trusted and powerful principal is $⊤$
(where $⊤≽p$ for all $p$). We also assume there is a set of
_atomic principals_ like "alice" and "bob" that define their own delegations.

Since the trust ordering defines a lattice, it has meet and join
operations.
Principal $p∧q$ is the least powerful principal that can
act for both $p$ and $q$; conversely, $p∨q$ can act for all
principals that both $p$ and $q$ can act for.
The least element in the information flow ordering is $⊤^{←}$,
representing maximal integrity and minimal confidentiality, whereas
the greatest element is $⊤^{→}$, representing minimal integrity and
maximal confidentiality. The join and meet operators in the
information flow lattice are the usual ⊔ and ⊓, respectively.

Any principal (label) can be expressed in a normal form $p^{→}∧q^{←}$
where $p$ and $q$ are CNF formulas over atomic principals~\cite{flam}.
This normal form allows us to decompose decisions about lattice ordering (in
either lattice) into separate questions regarding the integrity component ($p$)
and the confidentiality component ($q$).
Lattice operations can be similarly decomposed.

FLAM also introduces the concept of the _voice_ of a label (principal) $ℓ$,
written $\voice{ℓ}$.
Formally, for a normal-form label $ℓ = p^{→} ∧ q^{←}$, we define
voice as follows:
$\voice{p^{→} ∧ q^{←}} ~\triangleq~ p^{←}$.\footnote
{FLAM defines $\voice{p^{→} ∧ q^{←}} = p^{←} ∧ q^{←}$,
 but our simplified definition is sufficient for \nmlang.
 For clarity, the operator $\voice*$
 is always applied to a projected principal.}
A label's voice represents the minimum integrity needed to securely declassify
data constrained by that label, a restriction designed to enforce robust
declassification.

The Flow-Limited Authorization Calculus (FLAC)~\cite{flac} previously embedded
a simplified version of the FLAM proof system into a core language for
enforcing secure authorization and information flow.
FLAC is an extension of the Dependency Core Calculus (DCC)~\cite{ccd99,abadi06}
whose types contain FLAM labels.
A computation is additionally associated with a program-counter label $\pc$
which tracks the influences on the control flow and values that are not
explicitly labeled.

In this paper we take a similar approach: \nmlang enforces
security policies by performing computation in a monadic context.
As in FLAC, \nmlang includes a $\pc$ label.
For an ordinary value $v$, the monadic term $\return{ℓ}{v}$
signifies that value with the information flow label $ℓ$. If
value $v$ has type $τ$, the term $\return{ℓ}{v}$ has type
$\says{ℓ}{τ}$, capturing the confidentiality and integrity
of the information.

Unlike FLAC, \nmlang has no special support for dynamic
delegation of authority. Atomic principals define $\L$
by statically delegating their authority to arbitrary conjunctions and disjunctions of other
principals, and we include traditional
declassification and endorsement operations, "decl" and "endorse".  We leave to
future work the integration of nonmalleable information flow with
secure dynamic delegation.

\section{Enforcing nonmalleability}
\label{sec:typesys}

Multiple prior security-typed languages---both
functional~\cite{flac} and imperative~\cite{msz06,cm06,am11}---aim to allow
some form of secure downgrading.
These languages place no restriction whatsoever on the confidentiality of
endorsed data or the context in which an endorsement occurs.
Because of this permissiveness, all three insecure examples from
Section~\ref{sec:motivation} type-check in these languages.

\subsection{Robust declassification}

Robust declassification prevents adversaries from using
declassifications in the program to release information that was not
intended to be released.
The adversary is assumed to be able to
observe some state of the system, whose confidentiality label is
sufficiently low, and to modify some state of the system,
whose integrity label is sufficiently low. Semantically,
robust declassification says that if the attacker is unable to
learn a secret with one attack, no other attack will cause
it to be revealed~\cite{zm01b,msz06}.
The attacker has no control over information
release because all attacks are equally good.
When applied to a decentralized system, robust declassification
means that for any principal $p$, other principals that $p$ does not trust
cannot influence declassification of $p$'s secrets~\cite{cm06}.

To enforce robust declassification, prior security-typed
languages place integrity constraints on
declassification.
The original work on FLAM enforces robust declassification
using the voice operator $\voice*$. However, when declassification is
expressed as a programming-language operation, as is more typical, it
is convenient to define a new operator on labels, one that maps in the
other direction, from integrity to confidentiality.  We define the
_view_ of a principal as the upper bound on the confidentiality a label
or context can enforce to securely endorse that label:

\begin{definition}[Principal view]
  Let $ℓ = p^{→} \tjoin q^{←}$ be a FLAM label (principal) expressed in normal form.
  The _view_ of $ℓ$, written $\view{ℓ}$, is defined as
  $\view{p^{→} \tjoin q^{←}} ~\triangleq~ q^{→}$.
\end{definition}

When the confidentiality of a label $ℓ$ lies above the view of
its own integrity, a declassification of that label may give
adversaries the opportunity to subvert the declassification to
release information. Without enough integrity, an adversary might,
for example, replace the information that is intended to
be released via declassification with some other secret.

Figure~\ref{fig:robdecl} illustrates this idea graphically.  It
depicts the lattice of FLAM labels, which is a product lattice with
two axes, confidentiality and integrity. A given label $ℓ$ is a point
in this diagram, whereas the set of labels sharing the same
confidentiality $ℓ^{→}$  or integrity $ℓ^{←}$ correspond to lines
on the diagram. Given the integrity $ℓ^{←}$ of the label $ℓ$, the
view of that integrity, $\view{ℓ^{←}}$, defines a region of information
(shaded) that is too confidential to be declassified.

The view operator directly corresponds to the writers-to-readers
operator that \citet{cm06} use to enforce robust declassification in
the DLM. We generalize the same idea here to the more expressive
labels of FLAM.

\begin{figure}
\begin{center}
\includegraphics[width=20em]{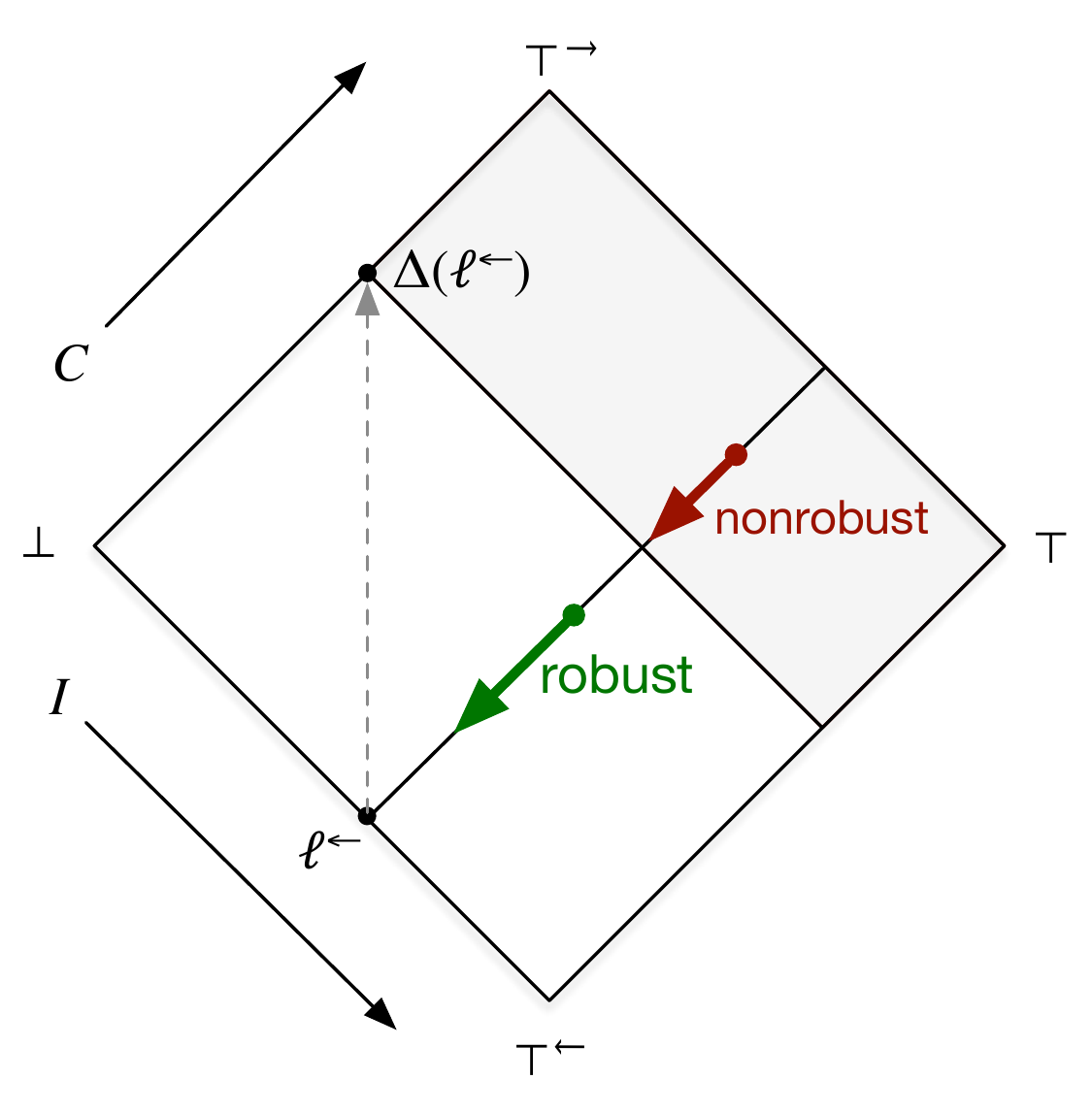}
\end{center}
\caption{Robust declassification says information at level $\ell$
can be declassified only if it has enough integrity.
The gray shaded region represents information that $\view{\integ{\ell}}$
cannot read, so it is unsafe to declassify with $\ell$'s integrity.}
\label{fig:robdecl}
\end{figure}

\subsection{Transparent endorsement}

The key insight of this work is that endorsement should be restricted
in a manner dual to robust declassification; declassification (reducing confidentiality)
requires a minimum integrity, so endorsement (raising integrity) should require
a _maximum_ confidentiality.
Intuitively, if a principal could have written data it cannot read, which we
call an ``opaque write,'' it is unsafe to endorse that data.
An endorsement is _transparent_ if it endorses only
information its authors could read.

The voice operator suffices to express this new restriction
conveniently, as depicted in Figure~\ref{fig:transparent}. In the
figure, we consider endorsing information with confidentiality
$ℓ^{→}$.  This confidentiality is mapped to a corresponding integrity
level $\voice{ℓ^{→}}$, defining a minimal integrity level that $ℓ$
must have in order to be endorsed. If $ℓ$ lies below this boundary,
its endorsement is considered transparent; if it lies above the
boundary, endorsement is _opaque_ and hence insecure.
The duality with robust declassification is clear.

\begin{figure}
\begin{center}
\includegraphics[width=20em]{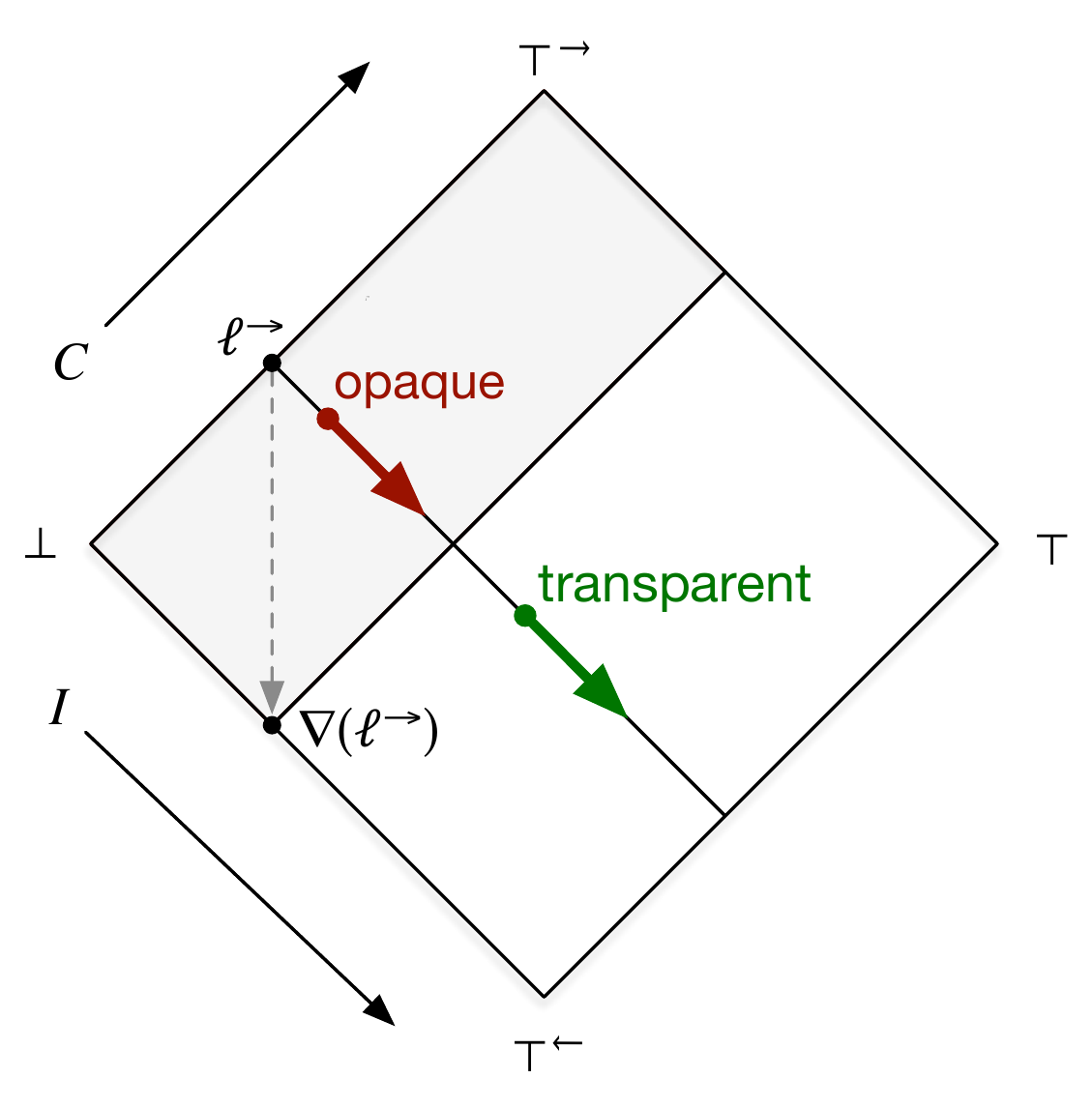}
\end{center}
\caption{Transparent endorsement in \nmlang.
The gray shaded region represents information that $\voice{\confid{\ell}}$ does
not trust and may have been created by an opaque write.
It is thus unsafe to endorse with $\ell$'s confidentiality.}
\label{fig:transparent}
\end{figure}

\section{A core language: \nmlang}
\label{sec:nmlang}

\begin{figure}
  \small
  \[
    \begin{array}{rcl}
      \multicolumn{2}{l}{ n ∈ \N}  & \text{\qquad (atomic principals)} \\
      \multicolumn{2}{l}{ x ∈ \mathcal{V}} & \text{\qquad (variable names)} \\
      \multicolumn{2}{l}{ π ∈ \{ →, ← \}} & \text{\qquad (security aspects)} \\
      \\
      p,ℓ,\pc &::=&  n \sep \top \sep \bot \sep p^{π} \sep p ∧ p\sep p ∨ p \sep p ⊔ p \sep p ⊓ p \\[0.8em]
      τ &::=& \voidtype \sep \func{τ}{\pc}{τ} \sep \says{ℓ}{τ} \\[0.8em]
      v &::=& \void \sep \lamc{x}{τ}{\pc}{e} \sep \vreturn{ℓ}{v} \\[0.8em]
      e &::=& x \sep v \sep e~e \sep \return{ℓ}{e} \sep \bind{x}{e}{e} \\[0.4em]
        & \sep & \cdowngrade{e}{ℓ} \sep \idowngrade{e}{ℓ}
    \end{array}
  \]
\caption{Core \nmlang syntax.}
\label{fig:core-syntax}
\end{figure}
We now describe the NonMalleable Information Flow Calculus (\nmlang), a new
core language, modeled on DCC and FLAC, that allows downgrading, but in a more
constrained manner than FLAC so as to provide stronger semantic guarantees.
\nmlang incorporates the program-counter label $\pc$ of FLAC, but eschews the
more powerful "assume" mechanism of FLAC in favor of more traditional
declassify and endorse operations.

The full \nmlang is a small extension of Polymorphic DCC~\cite{abadi06}.
In Figure~\ref{fig:core-syntax} we present the core syntax, leaving
other features such as sums, pairs, and polymorphism to Appendix~\ref{sec:full-nmlang}.
Unlike DCC, \nmlang supports downgrading and models it as an effect. It is
necessary to track what information influences control flow so that
these downgrading effects may be appropriately constrained. Therefore,
like FLAC, \nmlang adds $\pc$ labels to lambda terms and types.

Similarly to DCC, protected values have type $\says{ℓ}{τ}$ where $ℓ$
is the confidentiality and integrity of a value of type $τ$.  All
computation on these values occurs in the "says" monad; protected
values must be bound using the "bind" term before performing operations
on them (e.g., applying them as functions).
Results of such computations are protected with the monadic unit
operator $\return{ℓ}{e}$, which protects the result of $e$ with label
$ℓ$.

\begin{figure}
\begin{flushleft}
  \rulefiguresize
  \boxed{e \stepsone e'}
  \begin{mathpar}
    \erule{E-App}{}{(\lamc{x}{τ}{\pc}{e})~v}{\subst{e}{x}{v}}

    \erule{E-BindM}{}{\bind{x}{\vreturn{ℓ}{v}}{e}}{\subst{e}{x}{v}}
    \hfill
  \end{mathpar}

  {\small
  \begin{align*}
  \text{(event)} \quad c & ::=   •\sep\vreturn{ℓ}{v} \sep \downemit*{π}{v} \\
  \text{(trace)} \quad t & ::= ε \sep c \sep t ; t
  \end{align*}
  }

  \boxed{\evalctx*{e} \stepsonectx \evalctx'{e'}}
  \begin{mathpar}
    \erulectx{E-Step}{e \stepsone e'}{\evalctx*{e}}{\evalctx{e'}{t;\bullet}}

    \erulectx{E-UnitM}{}{\evalctx*{\return{ℓ}{v}}}{\evalctx{\vreturn{ℓ}{v}}{t;\vreturn{ℓ}{v}}}

    \erulectx{E-Decl}{}{\evalctx*{\cdowngrade{\vreturn{ℓ'}{v}}{ℓ}}}{\evalctx{\vreturn{ℓ}{v}}{t;\downemit*{→}{v}}}

    \erulectx{E-Endorse}{}{\evalctx*{\idowngrade{\vreturn{ℓ'}{v}}{ℓ}}}{\evalctx{\vreturn{ℓ}{v}}{t;\downemit*{←}{v}}}

    \erulectx{E-Eval}{\evalctx*{e} \stepsonectx \evalctx'{e'}}{\evalctx*{E[e]}}{\evalctx'{E[e']}}
    \hfill
  \end{mathpar}

  \vspace{1em}
  \underline{\small Evaluation context}
  \[
    \small
    \begin{array}{rcl}
      E & ::= & [\cdot] \sep E~e \sep v~E \sep \return{ℓ}{E} \sep \bind{x}{E}{e} \\[0.4em]
        & \sep & \cdowngrade{E}{ℓ} \sep \idowngrade{E}{ℓ}
    \end{array}
  \]
\end{flushleft}
\caption{Core \nmlang operational semantics.}
\label{fig:core-semantics}
\end{figure}

\subsection{\nmlang operational semantics}

The core semantics of \nmlang are mostly standard, but to obtain our
theoretical results we need additional information
about evaluation.
This information is necessary because we want to identify, for
instance, whether information is ever available to an attacker
during evaluation, even if it is discarded and does not influence the
final result.  This approach gives an attacker more power; an attacker 
can see information at its level even if it is not output by the program.

The \nmlang semantics, presented in Figure~\ref{fig:core-semantics},
maintain a trace $t$ of events.  An event is emitted into the trace whenever a new
protected value is created and whenever a declassification or
endorsement occurs.  These events track
the observations or influence an attacker may have during
a run of an \nmlang program.
Formally, a trace can be an empty trace $ε$, a single event $c$, or
the concatenation of two traces with the associative operator ``;''
with identity $ε$.

When a source-level unit term $\return{ℓ}{v}$ is evaluated (rule
\ruleref{E-UnitM}), an event $\vreturn{ℓ}{v}$ is added to the trace
indicating that the value $v$ became protected at $ℓ$.
When a protected value is declassified, a declassification event
$\downemit*{→}{v}$ is emitted, indicating that $v$ was declassified
from $ℓ'$ to $ℓ$.  Likewise, an endorsement event $\downemit*{←}{v}$
is emitted for an endorsement.  Other evaluation steps (rule
\ruleref{E-Step}) emit $•$, for ``no event.'' Rule \ruleref{E-Eval} steps under the evaluation 
contexts~\cite{wf94} defined at the bottom of Figure~\ref{fig:core-semantics}.

Rather than being literal side effects of the program, these events track
how observable information is as it is accessed, processed, and
protected by the program.  Because our semantics emits an event
whenever information is protected (by evaluating an η term) or
downgraded (by a "decl" or "endorse" term), our traces capture all
information processed by a program, indexed by the policy protecting
that information.

By analogy, these events are similar to the typed and labeled mutable
reference cells of languages like FlowCaml~\cite{ps03} and
DynSec~\cite{zm07}.  An event $\vreturn{ℓ}{v}$ is analogous to
allocating a reference cell protected at $ℓ$, and $\downemit*{π}{v}$
is analogous to copying the contents of a cell at $ℓ'$ to a new cell
at $ℓ$.

It is important for the semantics to keep track of these events so
that our security conditions hold for programs containing data
structures and higher-order functions. Previous language-based
definitions of robust declassification have only
applied to simple "while"-languages~\cite{msz06,cm06,am11} or to primitive
types~\cite{flac}.

\subsection{\nmlang type system}
\label{sec:type-system}

\begin{figure}
\begin{flushleft}
  \rulefiguresize
  \boxed{\protjudge*{ℓ}{τ}}
  \begin{mathpar}
    \protrule{P-Unit}{}{ℓ}{\voidtype}

    \protrule{P-Lbl}{\stflowjudge*{ℓ'}{ℓ}}{ℓ'}{\says{ℓ}{τ}}
    \hfill
  \end{mathpar}
\end{flushleft}
\caption{Type protection levels.}
\label{fig:core-protect}
\end{figure}

The \nmlang protection relation, presented in
Figure~\ref{fig:core-protect}, defines how types relate to information
flow policies.  A type $τ$ protects the confidentiality and integrity
of $ℓ$ if $\protjudge*{ℓ}{τ}$.  Unlike in DCC and
FLAC, a label is protected by a type only if it flows to the
outermost "says" principal.  In FLAC and DCC, the types
$\says{ℓ'}{\says{ℓ}{τ}}$ and $\says{ℓ}{\says{ℓ'}{τ}}$ protect the same
set of principals; in other words, "says" is commutative.  By
distinguishing between these types, \nmlang does not provide the same
commutativity. 

The commutativity of "says" is a design decision, offering
a more permissive programming model at the cost of less
precise tracking of dependencies.  \nmlang takes advantage of this
extra precision in the \ruleref{UnitM} typing rule so the label
on every η term protects the information it contains, even if nested
within other η terms.
Abadi~\cite{ccd08} similarly modifies DCC's protection
relation to distinguish the protection level of terms with nested "says" types. 

\begin{figure}
\begin{flushleft}
  \rulefiguresize
  \boxed{\TValGpc{e}{τ}}
  \begin{mathpar}
    \Rule{Var}{}{\TVal{Γ,x\ty τ,Γ';\pc}{x}{τ}}

    \Rule{Unit}{}{\TValGpc{\void}{\voidtype}}

    \Rule{Lam}{%
      \TVal{Γ,x\ty τ_1;\pc'}{e}{τ_2}
    }{\TValGpc{\lamc{x}{τ_1}{\pc'}{e}}{\func{τ_1}{\pc'}{τ_2}}}

    \Rule{App}{%
      \TValGpc{e_1}{\func{τ'}{\pc'}{τ}} \\\\
      \TValGpc{e_2}{τ'} \\
      \stflowjudge*{\pc}{\pc'}
    }{\TValGpc{e_1~e_2}{τ}}

    \Rule{UnitM}{%
      \TValGpc{e}{τ} \\
      \stflowjudge*{\pc}{ℓ}
    }{\TValGpc{\return{ℓ}{e}}{\says{ℓ}{τ}}}

    \Rule{VUnitM}{
      \TValGpc{v}{τ}
    }{\TValGpc{\vreturn{ℓ}{v}}{\says{ℓ}{τ}}}

    \Rule{BindM}{%
      \TValGpc{e}{\says{ℓ}{τ'}} \\
      \protjudge*{ℓ}{τ} \\\\
      \TVal{Γ,x\ty τ';\pc ⊔ ℓ}{e'}{τ} \\
    }{\TValGpc{\bind{x}{e}{e'}}{τ}}

    \Rule{Decl}{%
      \TValGpc{e}{\says{ℓ'}{τ}} \\
      ℓ'^{←} = ℓ^{←} \\
      \stflowjudge*{\pc}{ℓ} \\\\
      \stflowjudge*{ℓ'^{→}}{ℓ^{→} ⊔ \view{(ℓ' ⊔ \pc)^{←}}} \\
    }{\TValGpc{\cdowngrade{e}{ℓ}}{\says{ℓ}{τ}}}

    \Rule{Endorse}{%
      \TValGpc{e}{\says{ℓ'}{τ}} \\
      ℓ'^{→} = ℓ^{→} \\
      \stflowjudge*{\pc}{ℓ} \\\\
      \stflowjudge*{ℓ'^{←}}{ℓ^{←} ⊔ \voice{(ℓ' ⊔ \pc)^{→}}} \\
    }{\TValGpc{\idowngrade{e}{ℓ}}{\says{ℓ}{τ}}}
    \hfill
  \end{mathpar}
\end{flushleft}
\caption{Typing rules for core \nmlang.}
\label{fig:core-types}
\end{figure}
The core type system presented in Figure~\ref{fig:core-types} enforces
nonmalleable information flow for \nmlang programs.  Most of the
typing rules are standard, and differ only superficially from DCC and
FLAC.  Like in FLAC, \nmlang typing judgments include a program
counter label, \pc, that represents an upper bound on the
confidentiality and integrity of bound information that any
computation may depend upon.  For instance, rule \ruleref{BindM}
requires the type of the body of a "bind" term to protect the
unlabeled information of type $τ'$ with at least $ℓ$, and to type-check
under a raised program counter label $\pc ⊔ ℓ$.  Rule
\ruleref{Lam} ensures that function bodies type-check with respect to
the function's $\pc$ annotation, and rule \ruleref{App} ensures
functions are only applied in contexts that flow to these annotations.

The \nmlang rule for \ruleref{UnitM} differs from FLAC and DCC
in
requiring the premise $\stflowjudge*{\pc}{ℓ}$ for well-typed $η$
terms.  This premise ensures a more precise relationship between the
$\pc$ and $η$ terms.  Intuitively this restriction makes sense. The $\pc$
is a bound on all unlabeled information in the context.  Since an expression
$e$ protected with $\return{ℓ}{e}$ may depend on any of this
information, it makes sense to require that $\pc$ flow to $ℓ$.%
\footnote{The premise is not required in FLAC because protection is commutative.
For example, in a FLAC term such as $\bind{x}{v}{\return{ℓ'}{\return{ℓ}{x}}}$, $x$ 
may be protected by $ℓ$ or $ℓ'$.}

By itself, this restrictive premise would prevent public data from
flowing through secret contexts and trusted data from flowing through untrusted contexts. 
To allow such flows, we distinguish source-level $\return{ℓ}{e}$ terms from
run-time values $\vreturn{ℓ}{v}$, which have been fully evaluated.
These terms are only created by the operational semantics during evaluation and
no longer depend on the context in which they appear; they are closed terms.
Thus it is appropriate to omit the new premise in \ruleref{VUnitM}.
This approach allows us to require more precise flow tracking for the
explicit dependencies of protected expressions without restricting
where these values flow once they are fully evaluated.

Rule \ruleref{Decl} ensures a declassification from label $ℓ'$ to $ℓ$
is robust.
We first require $ℓ'^{←} = ℓ^{←}$ to ensure that this does not perform
endorsement. A more permissive premise $\stflowjudge*{ℓ'^{←}}{ℓ^{←}}$ is admissible, 
but requiring equality simplifies our proofs and does not reduce expressiveness since the declassification
can be followed by a subsequent relabeling.
The premise $\stflowjudge*{\pc}{ℓ}$ requires that
declassifications occur in high-integrity contexts, and prevents
declassification events from creating implicit flows.
The premise $\stflowjudge{ℓ'^{→}}{ℓ^{→} ⊔ \view{(ℓ' ⊔ \pc)^{←}}}$ ensures that
the confidentiality of the information declassified does not exceed the view of
the integrity of the principals that may have influenced it.
These influences can be either explicit ($ℓ'^{←}$) or implicit ($\pc^{←}$), so
we compare against the join of the two.\footnote{
  The first two premises---$ℓ'^{←} = ℓ^{←}$ and $\stflowjudge*{\pc}{ℓ}$---make
  this join redundant. It would, however, be necessary if we replaced the
  equality premise with the more permissive $\stflowjudge*{ℓ'^{←}}{ℓ^{←}}$
  version, so we include it for clarity.}
This last premise effectively combines the two conditions identified by
\citet{cm06} for enforcing robust declassification in an imperative "while"-language.

Rule \ruleref{Endorse} enforces transparent endorsement.
All but the last premise are straightforward: the expression does not
declassify and $\stflowjudge*{\pc}{ℓ}$ requires a high-integrity context to
endorse and prevents implicit flows.
Interestingly, the last premise is dual to that in \ruleref{Decl}.
An endorsement cannot raise integrity above the voice of the confidentiality of
the data being endorsed ($ℓ'^{→}$) or the context performing the endorsement
($\pc^{→}$).
For the same reasons as in \ruleref{Decl}, we compare against their join.

\subsection{Examples revisited}

We now reexamine the examples presented in Section~\ref{sec:motivation}
to see that the \nmlang type system prevents the vulnerabilities seen above.

\subsubsection{Password checker}

We saw above that when the password checker labels "guess" at $T^{→}$,
well-typed code can improperly set "guess" to the actual password.
We noted that the endorsement enabled an insecure flow of information.
Looking at \ruleref{Endorse} in \nmlang, we can attempt to type the equivalent
expression: $\idowngrade{\it guess}{T}$.
However, if _guess_ has type $\says{T^{→}}{"bool"}$, the "endorse" does not type-check;
it fails to satisfy the final premise of \ruleref{Endorse}:
\[ \notstflowjudge*{⊥^{←} = (T^{→})^{←}}{T^{←} ⊔ \voice{T^{→}} = T^{←}}. \]
If we instead give "guess" the label $U^{←}$, the endorsement
type-checks, assuming a sufficiently trusted $\pc$.

This is as it should be.
With the label $U^{←}$, the guesser must be able to read their own guess,
guaranteeing that they cannot guess the correct password unless they in fact
know the correct password.
Figure~\ref{fig:sec-pwd} shows this secure version of the password checker.

\begin{figure}
\small
\begin{align*}
  & \texttt{checkpwd} = \lamcs{g \ty \says{U^{←}}{"String"}, p \ty \says{T}{"String"}}{T^{←}}{} \\
  & \quad \bind{\it guess}{(\idowngrade{g}{T^{←}})}{} \\
  & \qquad \cdowngrade{(\bind{\it pwd}{p}{{\return{T}{\mathit{pwd} == \mathit{guess}}}})}{T^{←}}
\end{align*}
\caption{A secure version of a password checker.}
\label{fig:sec-pwd}
\end{figure}

\subsubsection{Sealed-bid auction}

In the insecure auction described in Section~\ref{sec:cheat-auction}, we argued
that an insecure flow was created when $T$ endorsed "b_bid" from
$B^{←}∧(A∧B)^{→}$ to $A∧B$.
This endorsement requires a term of the form $\idowngrade{v}{A∧B}$ where $v$
types to $\says{B^{←}∧(A∧B)^{→}}{"int"}$.
Despite the trusted context, the last premise of \ruleref{Endorse} again fails:
\[ \notstflowjudge*{B^{←}}{(A∧B)^{←} ⊔ \voice{(A∧B)^{→}} = (A∧B)^{←}}. \]
If we instead label $"a\_bid" : \says{A}{"int"}$ and $"b\_bid" :
\says{B}{"int"}$, then the corresponding "endorse" statements
type-check, assuming that $T$ is trusted: $T ⊑ (A∧B)^{←}$.

\subsubsection{Laundering secrets}

For the secret-laundering example in Section~\ref{sec:launder-sec}, we assume
that neither $H$ nor $L$ is trusted, but the output from the program is.
This forces an implicit endorsement of $y$, the input received from $H$.
But the condition needed to endorse from $H^{→} ∧ ⊥^{←}$ to $H^{→} ∧
⊤^{←}$ is false:
\[ \stflowjudge*{⊥^{←}}{⊤^{←} ⊔ \voice{H^{→}} = \voice{H^{→}}} \]
We have $\notstflowjudge*{\voice{L^{→}}}{\voice{H^{→}}}$ and all integrity flows to
$⊥^{←}$, so by transitivity the above condition cannot hold.

\section{Security conditions}
\label{sec:condition}

The \nmlang typing rules enforce several strong security conditions:
multiple forms of conditional noninterference, robust
declassification, and our new transparent endorsement and nonmalleable
information flow conditions.  We define these conditions formally
but relegate proof details to Appendix~\ref{sec:proofs}.

\subsection{Attackers}
\label{sec:attacker}

Noninterference is frequently stated with respect to a specific but arbitrary label.
Anything below that label in the lattice is ``low'' (public or trusted)
and everything else is ``high''.
We broaden this definition slightly and designate high information
using a set of labels $\H$ that is upward closed.
That is, if $ℓ ∈ \H$ and $\stflowjudge*{ℓ}{ℓ'}$, then $ℓ' ∈ \H$.
We refer to such upward closed sets as _high sets_.

We say that a type $τ$ is a _high type_, written ``$\hightype{\H}{τ}$'',
if all of the information in a value of type $τ$ is above some label
in the high set $\H$. The following rule defines high types:
\begin{mathpar}
  \Rule[$\H$ is upward closed]{P-Set}{%
    H ∈ \H \\
    \protjudge*{H}{τ}
  }{\hightype{\H}{τ}}
  \hfill
\end{mathpar}

This formulation of adversarial power is adequate to express noninterference, in which
confidentiality and integrity do not interact. However,
our more complex conditions relate confidentiality to integrity and
therefore require a way to relate the attacker's power in the two domains.

Intuitively, an attacker is an arbitrary set of colluding atomic principals.
Specifically, if $n_1, \dotsc, n_k ∈ \N$ are those atomic principals, then the
set $\A = \{ ℓ ∈ \L \mid \stafjudge*{n_1 ∧ \dotsb ∧ n_k}{ℓ} \}$ represents this
attacker's power. These principals may include principals mentioned in
the program, and there may be delegations between attacker
principals and program principals.
While this definition of an attacker is intuitive, the results in
this paper actually hold for a more general notion of attacker defined in
\cref{sec:attacker-properties}.

Attackers correspond to two high sets: an _untrusted_ high set
$\U = \{ ℓ ∈ \L \mid ℓ^{←} ∈ \A \}$ and a _secret_ high set
$\secret = \{ ℓ ∈ \L \mid ℓ^{→} \notin \A \}$.
We say that $\A$ _induces_ $\U$ and $\secret$.

\subsection{Equivalences}

All of our security conditions involve relations on traces.  As is
typically the case for information-flow security conditions, we define
a notion of ``low equivalence'' on traces, which ignores effects with
high labels.  We proceed by defining low-equivalent expressions and
then extending low-equivalence to traces.

For expression equivalence, we examine precisely the values which are visible
to a low observer defined by a set of labels $\low$: $\vreturn{ℓ}{v}$ and
$\downemit*{π}{v}$ where $ℓ ∈ \low$.
We formalize this idea in Figure~\ref{fig:equivalence}, using $\bullet$ to represent
values that are not visible.
Beyond ignoring values unable to affect the output, we use a standard
structural congruence (i.e., syntactic equivalence).
This strict notion of equivalence is not entirely necessary;
observational equivalence or any refinement thereof would be sufficient if
augmented with the $\bullet$-equivalences in Figure~\ref{fig:equivalence}.

\begin{figure}
\begin{flushleft}
  \rulefiguresize
  \boxed{\seteq*{c}{c'}} \boxed{\seteq*{v}{v'}} \\[0.5em]
  {\small
  These equivalence relations are the smallest
  congruences over $c$ and over $v$ extended with •,
  containing the equivalences defined by these rules:
  }
  \begin{mathpar}
    \Rule{Eq-UnitM}{ℓ \notin \low}{\seteq*{\vreturn{ℓ}{v}}{\bullet}}

    \Rule{Eq-Down}{ℓ \notin \low}{\seteq*{\downemit*{π}{v}}{\bullet}}
    \hfill
  \end{mathpar}

  \boxed{\traceeq*{t}{t'}}\\[0.5em]
  {\small
  The equivalence relation $\traceeq*{t}{t'}$ is the smallest
  congruence over $t$ containing the equivalences defined by these
  rules:
  }
  \begin{mathpar}
    \Rule{T-Lift}{\seteq*{c}{c'}}{\traceeq*{c}{c'}}

    \Rule{T-BulletR}{}{\traceeq*{t;\bullet}{t}}

    \Rule{T-BulletL}{}{\traceeq*{\bullet;t}{t}}
  \end{mathpar}
\end{flushleft}
\caption{Low equivalence and low trace equivalence.}
\label{fig:equivalence}
\end{figure}

Figure~\ref{fig:equivalence} also extends the equivalence on emitted values to
equivalence on entire traces of emitted values. Essentially, two
traces are equivalent if there is a way to match up equivalent events
in each trace, while ignoring high events equivalent to •.

\subsection{Noninterference and downgrading}
\label{sec:ni}

An immediate consideration when formalizing information flow is how
to express interactions between an adversary and the system.
One possibility is to limit interaction to inputs and outputs of the program.
This is a common approach for functional
languages. We take a stronger approach in which security
is expressed in terms of execution traces. Note that traces
contain all information necessary to ensure the security of
input and output values.

We begin with a statement of noninterference in the presence of
downgrading.
Theorem~\ref{thm:no-downgrade-ni} states that, given two high inputs, a
well-typed program produces two traces that are either low-equivalent or
contain a downgrade event that distinguishes them.
This implies that
differences in traces distinguishable by an attacker are all
attributable to downgrades of information derived from the high inputs.
Furthermore, any program that performs no downgrades on secret or untrusted values
(i.e., contain no "decl" or "endorse" terms on $\H$ data) must be noninterfering.

\begin{theorem}[Noninterference modulo downgrading]
  \label{thm:no-downgrade-ni}
  Let $\H$ be a high set and let $\low = \L \setminus \H$.
  Given an expression $e$ such that $\TVal{Γ,x\ty τ_1;\pc}{e}{τ_2}$ where
  $\hightype{\H}{τ_1}$, for all $v_1, v_2$ with $\TValGpc{v_i}{τ_1}$, if
  \[ \evalctx{\subst{e}{x}{v_i}}{v_i} \stepstoctx \evalctx{v'_i}{t^i} \]
  then either there is some event $\downemit*{π}{w} ∈ t^i$ where $ℓ' ∈ \H$ and
  $ℓ \notin \H$, or $\traceeq*{t^1}{t^2}$.
\end{theorem}

The restrictions placed on downgrading operations mean that we can
characterize the conditions under which no downgrading can occur.
We add two further noninterference theorems that restrict downgrading in different ways.
Theorem~\ref{thm:high-pc-ni} states that if a program types without a public–trusted $\pc$
it must be noninterfering (with respect to that definition of ``public–trusted'').

\begin{theorem}[Noninterference of high-$\pc$ programs]
  \label{thm:high-pc-ni}
  Let $\A$ be an attacker inducing high sets $\U$ and $\secret$.
  Let $\H$ be one of those high sets and $\low = \L \setminus \H$.
  Given some $e$ such that $\TVal{Γ,x\ty τ_1;\pc}{e}{τ_2}$ where
  $\hightype{\H}{τ_1}$, for all $v_1, v_2$ with $\TValGpc{v_i}{τ_1}$, if
  $\evalctx{\subst{e}{x}{v_i}}{v_i} \stepstoctx \evalctx{v'_i}{t^i}$
  and $\pc ∈ \U ∪ \secret$, then $\traceeq*{t^1}{t^2}$.
\end{theorem}

Rather than restrict the $\pc$, Theorem~\ref{thm:secret-untrust-ni} states that secret–untrusted
information is _always_ noninterfering.  Previous work (e.g., \cite{msz06,am11}) does not
restrict endorsement of confidential information, allowing any label to be downgraded to
public–trusted (given a public–trusted $\pc$).  In \nmlang, however,
secret–untrusted data must remain secret and untrusted.

\begin{theorem}[Noninterference of secret–untrusted data]
  \label{thm:secret-untrust-ni}
  Let $\A$ be an attacker inducing high sets $\U$ and $\secret$.
  Let $\H = \U ∩ \secret$ and $\low = \L \setminus \H$.
  Given some $e$ such that $\TVal{Γ,x\ty τ_1;\pc}{e}{τ_2}$ where
  $\hightype{\H}{τ_1}$, for all $v_1, v_2$ with $\TValGpc{v_i}{τ_1}$, if
  $\evalctx{\subst{e}{x}{v_i}}{v_i} \stepstoctx \evalctx{v'_i}{t^i}$
  then $\traceeq*{t^1}{t^2}$.
\end{theorem}

\subsection{Robust declassification and irrelevant~inputs}
\label{sec:irrel-inputs}
\label{sec:rob-decl-thm}

We now move to security conditions for programs that do not satisfy noninterference.
Recall that robust declassification informally means the attacker has
no influence on what information is released by declassification.
Traditionally, it is stated in terms of attacker-provided
code that is inserted into low-integrity 
holes in programs which
differ only in their secret inputs. In \nmlang,
the same attacker power can be obtained by substituting
exactly two input values into the program, one secret and one
untrusted. This simplification is possible because \nmlang has
first-class functions that can model the substitution of low-integrity code.
Appendix~\ref{sec:nmif-generalize} shows that this simpler two-input definition
is equivalent to the traditional hole-based approach in the full
version of \nmlang (Appendix~\ref{sec:full-nmlang}).

Prior work on "while"-based languages~\cite{msz06,cm06} define robust
declassification in terms of four traces generated by the combination of two variations: 
a secret input and some attacker-supplied code.
For terminating traces, these definitions require any pair of secrets to 
produce public-equivalent traces under all attacks or otherwise to produce distinguishable traces
regardless of the attacks chosen.  This implies that an attacker cannot control the disclosure of secrets.

We can attempt to capture this notion of robust declassifcation using the notation of
\nmlang.
For a program $e$ with a secret input $x$ and untrusted input $y$, we wish to
say $e$ robustly declassifies if, for all secret values $v_1, v_2$ and for all
untrusted values $w_1, w_2$, where
\[ \evalctx{\subst{\subst{e}{x}{v_i}}{y}{w_j}}{v_i;w_j} \stepstoctx \evalctx{v_{ij}}{t^{ij}}, \]
then $\traceeq{\public}{t^{11}}{t^{21}} ~ \Longleftrightarrow ~ \traceeq{\public}{t^{12}}{t^{22}}$.

This condition is intuitive but, unfortunately, overly restrictive.
It does not account for the possibility of an _inept attack_, in which an
attacker causes a program to reveal less information than intended.

Inept attacks are harder to characterize than in previous work because,
unlike the previously used "while"-languages, \nmlang supports
data structures with heterogeneous labels.
Using such data structures, we can build a program that implicitly
declassifies data by using a secret to influence the selection of an
attacker-provided value and then declassifying that selection.
Figure~\ref{fig:inept-attack} provides an example of such a program,
which uses sums and products from the full \nmlang language.

\begin{figure}
  \small
  \addtolength\jot{-2pt}
  \begin{align*}
     & \big(λ(x:(\says{P^{→}∧U^{←}}{τ})\times(\says{P^{→}∧U^{←}}{τ}))[P^{→}∧T^{←}].  \\
     & \quad "decl"~\big("bind"~b = \return{S^{→}∧T^{←}}{\mathit{sec}}~"in" \\
     & \qquad\qquad "case"~b~"of"~\ione{_}.\return{S^{→}∧T^{←}}{(\pone{x})} \\
     & \qquad\qquad\hspace{3.57em} \mid \itwo{_}.\return{S^{→}∧T^{←}}{(\ptwo{x})}\big) \\
     & \qquad\quad "to"~P^{→}∧T^{←} \big) ~ \big\langle\mathit{atk}_1, \mathit{atk}_2\big\rangle
  \end{align*}
  \caption{A program that admits inept attacks.
    Here $\stflowjudge*{P}{S}$ and $\stflowjudge*{T}{U}$, but not vice versa,
    so $\mathit{sec}$ is a secret boolean and
    $\pair{\mathit{atk}_1}{\mathit{atk}_2}$ form an untrusted pair of values.
    If $\mathit{atk}_1 \neq \mathit{atk}_2$, then the attacker will learn the
    value of $\mathit{sec}$.
    If $\mathit{atk}_1 = \mathit{atk}_2$, however, then the attacker learns
    nothing due to its own ineptness.}
  \label{fig:inept-attack}
\end{figure}

While this program appears secure---the attacker has no control over what
information is declassified or when a declassification occurs---it violates
the above condition.
One attack can contain the same value twice---causing any two secrets to
produce indistinguishable traces---while the other can contain different
values.
Intuitively, no vulnerability in the program is thereby revealed;
the program was _intended_ to release information, but the attacker
failed to infer it due to a poor choice of attack.
Such inputs result in less information leakage entirely due to the attacker's
ineptness, not an insecurity of the program.
As a result, we consider inputs from inept attackers to be _irrelevant_ to our
security conditions.

Dually to inept attackers, we can define uninteresting secret inputs.
For example, if a program endorses an attacker's selection of a secret value,
an input where all secret options contain the same data is uninteresting, so we
also consider it irrelevant.

Which inputs are irrelevant is specific to the program and to the choice of attacker.
In Figure~\ref{fig:inept-attack}, if both execution paths used $(\pone{x})$,
there would be no way for an attacker to learn any information, so all attacks
are equally relevant.
Similarly, if $S^{→}$ is already considered public, then there is no secret
information in the first place, so again, all attacks are equally relevant.

For an input to be irrelevant,
it must have no
influence over the outermost layer of the data structure---the label that is
explicitly downgraded.
If the input could influence that outer layer in any way, the internal data
could be an integral part of an insecure execution.
Conversely, when the selection of nested values is independent of any
untrusted/secret information (though the content of the values may not be), it
is reasonable to assume that the inputs will be selected so that different
choices yield different results.
An input which does not is either an inept attack---an attacker gaining less
information than it could have---or an uninteresting secret---a choice between
secrets that are actually the same.
In either case, the input is irrelevant.

To ensure that we only consider data structures with nested values that were
selected independently of the values themselves, we leverage the
noninterference theorems in Section~\ref{sec:ni}.
In particular, if the outermost label is trusted before a declassification (or
public prior to an endorsement), then any influence from untrusted (secret) data
must be the result of a prior explicit downgrade.
Thus we can identify irrelevant inputs by finding inputs that result in traces
that are public-trusted equivalent, but can be made both public (trusted)
equivalent and non-equivalent at the point of declassification (endorsement).

To define this formally, we begin by partitioning the principal lattice into
four quadrants using the definition of an attacker from
Section~\ref{sec:attacker}.
We consider only flows between quadrants and, as with noninterference,
downgrades must result in public or trusted values.
We additionally need to refer to individual elements and prefixes of traces.
For a trace $t$, let $\trelt{t}{n}$ denote the $n$th element of $t$, and let
$\trpref{t}{n}$ denote the prefix of $t$ containing its first $n$ elements.

\begin{definition}[Irrelevant inputs]
  \label{def:irrel-input}
  Consider attacker $\A$ inducing high sets $\H_{←}$ and $\H_{→}$. Let
  $\low_{π} = \L \setminus \H_{π}$ and $\low = \low_{←} ∩ \low_{→}$.
  Given opposite projections $π$ and $π'$ a program $e$, and types $τ_x$ and
  $τ_y$ such that $\hightype{\H_{π}}{τ_x}$ and $\hightype{\H_{π'}}{τ_y}$,
  we say an input $v_1$ is an \emph{irrelevant $π'$-input} with respect to $\A$ and $e$ if
  $\TValGpc{v_1}{τ_x}$ and there exist values $v_2$, $w_1$, and $w_2$ and
  four trace indices $n_{ij}$ (for $i, j ∈ \{1, 2\}$) such that the following conditions hold:
  \begin{enumerate}
    \item $\TValGpc{v_2}{τ_x}$, $\TValGpc{w_1}{τ_y}$, and $\TValGpc{w_2}{τ_y}$
    \item $\evalctx{\subst{\subst{e}{x}{v_i}}{y}{w_j}}{v_i;w_j} \stepstoctx \evalctx{v_{ij}}{t^{ij}}$
    \item $\nseteq*{\trelt{t^{ij}}{n_{ij}}}{\bullet}$ for all $i,j ∈ \{1,2\}$
    \item \label{ii:li:pt-equiv} $\traceeq*{\trpref{t^{ij}}{n_{ij}}}{\trpref{t^{kl}}{n_{kl}}}$ for all $i,j,k,l ∈ \{1,2\}$
    \item \label{ii:li:pair1-equiv} $\traceeq{\low_{π}}{\trpref{t^{11}}{n_{11}}}{\trpref{t^{12}}{n_{12}}}$
    \item \label{ii:li:pair2-nonequiv} $\ntraceeq{\low_{π}}{\trpref{t^{21}}{n_{21}}}{\trpref{t^{22}}{n_{22}}}$
  \end{enumerate}
  Otherwise we say $v_1$ is a \emph{relevant $π'$-input} with respect to $\A$ and $e$,
  denoted $\rel*{π'}{v_1}$. Note that the four indices $n_{ij}$ identify
  corresponding prefixes of the four traces.
\end{definition}

As mentioned above, prior downgrades can allow secret/untrusted information
to directly influence the outer later of the data structure, but
Condition~\ref{ii:li:pt-equiv} requires that all four trace prefixes be
public-trusted equivalent, so any such downgrades must have the same influence
across all executions.
Condition~\ref{ii:li:pair1-equiv} requires that some inputs result in prefixes
that are public equivalent (or trusted equivalent for endorsement), while
Condition~\ref{ii:li:pair2-nonequiv} requires that other inputs result in
prefixes that are distinguishable.
Since all prefixes are public-trusted equivalent, this means there is an
implicit downgrade inside a data structure, so the equivalent prefixes form an
irrelevant input.

We can now relax our definition of robust declassification to only
restrict the behavior of _relevant_ inputs.

\begin{definition}[Robust declassification]
  \label{def:rob-decl}
  Let $e$ be a program and let $x$ and $y$ be variables representing secret and
  untrusted inputs, respectively.
  We say that $e$ _robustly declassifies_ if, for all attackers $\A$ inducing
  high sets $\U$ and $\secret$ (and $\public = \L \setminus \secret$)
  and all values $v_1, v_2, w_1, w_2$, if
  \[ \evalctx{\subst{\subst{e}{x}{v_i}}{y}{w_j}}{v_i;w_j} \stepstoctx \evalctx{v_{ij}}{t^{ij}}, \]
  then $\left(\rel*{←}{w_1} \text{ and } \traceeq{\public}{t^{11}}{t^{21}}\right) ~ \Longrightarrow ~ \traceeq{\public}{t^{12}}{t^{22}}$.
\end{definition}

As \nmlang only restricts declassification of low-integrity data, endorsed data is
free to influence future declassifications.
As a result, we can only guarantee robust declassification in the absence of
endorsements.

\begin{theorem}[Robust declassification]
  \label{thm:rob-decl}
  Given a program $e$, if $\TVal{Γ,x\ty τ_x,y\ty τ_y;\pc}{e}{τ}$ and $e$
  contains no "endorse" expressions, then $e$ robustly declassifies as
  defined in Definition~\ref{def:rob-decl}.
\end{theorem}

Note that prior definitions of robust declassification~\cite{msz06,cm06}
similarly prohibit endorsement and ignore pathological inputs, specifically
nonterminating traces.
Our irrelevant inputs are very different since \nmlang is strongly
normalizing but admits complex data structures, but the need for some
restriction is not new.

\subsection{Transparent endorsement}
\label{sec:transp-endorse}

We described in Section~\ref{sec:motivation} how endorsing opaque
writes can create security vulnerabilities.
To formalize this intuition, we present _transparent endorsement_, a security
condition that is dual to robust declassification.
Instead of ensuring that untrusted information cannot meaningfully influence
declassification, transparent endorsement guarantees that
secret information cannot meaningfully influence endorsement.
This guarantee ensures that secrets cannot influence the endorsement of an attacker's
value---neither the value endorsed nor the existence of the endorsement itself.

As it is completely dual to robust declassification, we again appeal to
the notion of irrelevant inputs, this time to rule out uninteresting secrets.
The condition looks nearly identical, merely switching the roles of confidentiality
and integrity.
It therefore ensures that any choice of interesting secret provides an attacker
with the maximum possible ability to influence endorsed values; no
interesting secrets provide more power to attackers than others.

\begin{definition}[Transparent endorsement]
  \label{def:transp-endorse}
  Let $e$ be a program and let $x$ and $y$ be variables representing secret and
  untrusted inputs, respectively.
  We say that $e$ _transparently endorses_ if, for all attackers $\A$ inducing
  high sets $\U$ and $\secret$ (and $\trusted = \L \setminus \U$)
  and all values $v_1, v_2, w_1, w_2$, if
  \[ \evalctx{\subst{\subst{e}{x}{v_i}}{y}{w_j}}{v_i;w_j} \stepstoctx \evalctx{v_{ij}}{t^{ij}}, \]
  then $\left(\rel*{→}{v_1} \text{ and } \traceeq{\trusted}{t^{11}}{t^{12}}\right) ~ \Longrightarrow ~ \traceeq{\trusted}{t^{21}}{t^{22}}$.
\end{definition}

As in robust declassification, we can only guarantee transparent
endorsement in the absence of declassification.

\begin{theorem}[Transparent endorsement]
  \label{thm:transp-endorse}
  Given a program $e$, if $\TVal{Γ,x\ty τ_x,y\ty τ_y;\pc}{e}{τ}$ and $e$
  contains no "decl" expressions, then $e$ transparently endorses.
\end{theorem}

\subsection{Nonmalleable information flow}
\label{sec:nmif}

Robust declassification and transparent endorsement each restrict one type of
downgrading, but as structured above, cannot be enforced in the presence of
both declassification and endorsement.
The key difficulty stems from the fact that previously
declassified and endorsed data should be able to influence future
declassifications and endorsements.
However, any endorsement allows an attack to influence declassification, so
varying the secret input can cause the traces to deviate for one attack and not
another.
Similarly, once a declassification has occurred, we can say little about
the relation between trace pairs that fix a secret and vary an attack.

There is one condition that allows us to safely relate trace pairs even after a
downgrade event: if the downgraded values are identical in both trace pairs.
Even if a declassify or endorse could have caused the traces to deviate, if it
did not, then this program is essentially the same as one that started with
that value already downgraded and performed no downgrade.
To capture this intuition, we define nonmalleable information flow in terms of
trace prefixes that either do not deviate in public values when varying only
the secret input or do not deviate in trusted values when varying only the
untrusted input.
This assumption may seem strong at first, but it exactly captures the
intuition that downgraded data---but not secret/untrusted data---should be able
to influence future downgrades.
While two different endorsed attacks could influence a future declassification,
if the attacks are similar enough to result in the same value being endorsed,
they must influence the declassification _in the same way_.

\begin{definition}[Nonmalleable information flow]
  \label{def:nmifc}
  Let $e$ be a program and let $x$ and $y$ be variables representing secret and
  untrusted inputs, respectively.
  We say that $e$ enforces _nonmalleable information flow_ (NMIF) if the
  following holds for all attackers $\A$ inducing high sets $\U$ and $\secret$.
  Let $\trusted = \L \setminus \U$, $\public = \L \setminus \secret$ and
  $\low = \trusted ∩ \secret$.
  For all values $v_1$, $v_2$, $w_1$, and $w_2$, let
  \[ \evalctx{\subst{\subst{e}{x}{v_i}}{y}{w_j}}{v_i;w_j} \stepstoctx \evalctx{v_{ij}}{t^{ij}}. \]
  For all indices $n_{ij}$ such that $\nseteq{\low}{\trelt{t^{ij}}{n_{ij}}}{\bullet}$
  \begin{enumerate}[leftmargin=*]
    \item \label{li:nmifc:rd}
      If $\traceeq{\trusted}{\trpref{t^{i1}}{n_{i1} - 1}}{\trpref{t^{i2}}{n_{i2} - 1}}$ for
      $i = 1,2$, then
      \[ \left(\rel*{←}{w_1} \text{ and } \traceeq{\public}{\trpref{t^{11}}{n_{11}}}{\trpref{t^{21}}{n_{21}}}\right) ~ \Longrightarrow ~ \traceeq{\public}{\trpref{t^{12}}{n_{12}}}{\trpref{t^{22}}{n_{22}}}. \]
    \item \label{li:nmifc:te}
      Similarly, if $\traceeq{\public}{\trpref{t^{1j}}{n_{1j} - 1}}{\trpref{t^{2j}}{n_{2j} - 1}}$
      for $j = 1,2$, then
      \[ \left(\rel*{→}{v_1} \text{ and } \traceeq{\trusted}{\trpref{t^{11}}{n_{11}}}{\trpref{t^{12}}{n_{12}}}\right) ~ \Longrightarrow ~ \traceeq{\trusted}{\trpref{t^{21}}{n_{21}}}{\trpref{t^{22}}{n_{22}}}. \]
  \end{enumerate}
\end{definition}

Unlike the previous conditions, \nmlang enforces NMIF with no syntactic restrictions.

\begin{theorem}[Nonmalleable information flow]
  \label{thm:nmifc}
  For any program $e$ such that $\TVal{Γ,x\ty τ_x,y\ty τ_y;\pc}{e}{τ}$, $e$ enforces NMIF.
\end{theorem}

We note that both Theorems~\ref{thm:rob-decl} and \ref{thm:transp-endorse} are
directly implied by Theorem~\ref{thm:nmifc}.
For robust declassification, the syntactic prohibition on "endorse" directly
enforces $\traceeq{\trusted}{t^{i1}}{t^{i2}}$ (for the entire trace), and the
rest of case~\ref{li:nmifc:rd} is exactly that of Theorem~\ref{thm:rob-decl}.
Similarly, the syntactic prohibition on "decl" enforces
$\traceeq{\public}{t^{1j}}{t^{2j}}$, while the rest of case~\ref{li:nmifc:te}
is exactly Theorem~\ref{thm:transp-endorse}.

\section{NMIF as 4-safety}
\label{sec:hyperproperty}

\citet{cs08} define a _hyperproperty_ as ``a set of sets of infinite traces,''
and _hypersafety_ to be a hyperproperty that can be
characterized by a finite set of finite trace prefixes defining some ``bad thing.''
That is, given any of these finite sets of trace prefixes it is impossible to
extend those traces to satisfy the hyperproperty.
It is therefore possible to show that a program satisfies a hypersafety
property by proving that no set of finite trace prefixes emitted by the program
fall into this set of ``bad things.''
They further define a _$k$-safety hyperproperty_ (or _$k$-safety_) as a
hypersafety property that limits the set of traces needed to identify a
violation to size $k$.

Clarkson and Schneider note that noninterference provides an example of 2-safety.
We demonstrate here that robust declassification, transparent endorsement, and
nonmalleable information flow are all 4-safety properties.\footnote{While \nmlang produces finite traces 
and hyperproperties are defined for
infinite traces, we can easily extend \nmlang traces by stuttering $\bullet$
infinitely after termination.}

For a condition to be 2-safety, it must be possible to demonstrate a
violation using only two finite traces.
With noninterference, this demonstration is simple: if two traces with low-equivalent inputs
are distinguishable by a low observer, the program is interfering.

Robust declassification, however, cannot be represented this way.
It says that the program's confidentiality release events cannot be influence
by untrusted inputs.
If we could precisely identify the release events, this would allow us to
specify robust declassification as a 2-safety property on those release
events.
If every pair of untrusted inputs results in the same trace of confidentiality
release events, the program satisfies robust declassification.
However, to identify confidentiality release events requires comparing
traces with different secret inputs.
A trace consists of a set of observable states, not a set of release
events. Release events are identified by varying secrets; the robustness of releases is
identified by varying untrusted input.
Thus we need 4 traces to properly characterize robust
declassification.

Both prior work~\cite{cm06} and our definition in \cref{sec:rob-decl-thm} state
robust declassification in terms of four traces, making it easy to convert to a
4-hyperproperty.
That formulation cannot, however, be directly translated to 4-safety.
It instead requires a statement about trace prefixes, which cannot be
invalidated by extending traces.

Instead of simply reformulating Definition~\ref{def:rob-decl} with trace
prefixes, we modify it using insights gained from the definition of NMIF.
In particular, instead of a strict requirement that if a relevant attack
results in public-equivalent trace prefixes then other attacks must as well,
we relax this requirement to apply only when the trace prefixes are
trusted-equivalent.
As noted in Section~\ref{sec:nmif}, if we syntactically prohibit "endorse"---%
the only case in which we could enforce the previous definition---this
trivially reduces to that definition.
Without the syntactic restriction, however, the new condition is still enforceable.

For a given attacker $\A$ we can define a 4-safety property with respect to $\A$ (let
$\U$, $\secret$, $\trusted$, $\public$, and $\low$ be as in
Definition~\ref{def:nmifc}).
{\small
\begin{align*}
  \mathbf{RD}_{\A} ~\triangleq~ \Big\{ \mathbf{T} & ⊆ \mathbb{T} \mid \mathbf{T} = \left\{t^{11}, t^{12}, t^{21}, t^{22} \right\} \\
  & \land \trelt{t^{ij}}{1} ≠ \bullet ~\land~ \trelt{t^{ij}}{2} ≠ \bullet ~\land~ \trelt{t^{i1}}{1} = \trelt{t^{i2}}{1} ~\land~ \trelt{t^{1j}}{2} = \trelt{t^{2j}}{2} \\
  & \Longrightarrow \Big(\forall \{n_{ij}\} ⊆ \mathbb{N} : \big(\nseteq*{\trelt{t^{ij}}{n_{ij}}}{\bullet} ~\land~ \traceeq{\trusted}{\trpref{t^{i1}}{n_{i1}-1}}{\trpref{t^{i2}}{n_{i2}-1}} \\
  & \qquad\qquad \land~\traceeq{\public}{\trpref{t^{11}}{n_{11}}}{\trpref{t^{21}}{n_{21}}}~\land~\ntraceeq{\public}{\trpref{t^{12}}{n_{12}}}{\trpref{t^{22}}{n_{22}}}\big) \\
  & \qquad \Longrightarrow \seteq*{\trpref{t^{12}}{n_{12}}}{\trpref{t^{22}}{n_{22}}} \Big)\Big\}
\end{align*}
}
We then define robustness against all attackers as the
intersection over all attackers: $\mathbf{RD} = \bigcap_{\A} \mathbf{RD}_{\A}$.

The above definition structurally combines Definition~\ref{def:irrel-input}
with the first clause of Definition~\ref{def:nmifc} to capture both the
equivalence and the relevant-input statements of the original theorem.
In the nested implication, if the first two clauses hold
($\nseteq*{\trelt{t^{ij}}{n_{ij}}}{\bullet}$ and
$\traceeq{\trusted}{\trpref{t^{i1}}{n_{i1}-1}}{\trpref{t^{i2}}{n_{i2}-1}}$), then one of three
things must happen when fixing the attack and varying the secret: both trace
pairs are equivalent, both trace pairs are non-equivalent, or the postcondition
of the implication holds ($\seteq*{\trpref{t^{12}}{n_{12}}}{\trpref{t^{22}}{n_{22}}}$).
The first two satisfy the equivalency implication in Definition~\ref{def:nmifc}
while the third is exactly a demonstration that the first input is irrelevant.

Next we note that, while this does not strictly conform to the definition of
robust declassification in Definition~\ref{def:rob-decl} which cannot be
stated as a hypersafety property, $\mathbf{RD}$ is equivalent to Definition~\ref{def:rob-decl} 
for programs that do not perform endorsement.
This endorse-free condition means that the equivalence clause $\traceeq{\trusted}{\trpref{t^{i1}}{n_{i1}-1}}{\trpref{t^{i2}}{n_{i2}-1}}$ will be
true whenever the trace prefixes refer to the same point in execution.
In particular, they can refer to the end of execution, which gives exactly the
condition specified in the theorem.

As with every other result so far, the dual construction results in a 4-safety
property $\mathbf{TE}$ representing transparent endorsement.
Since $\mathbf{RD}$ captures the first clause of Definition~\ref{def:nmifc},
$\mathbf{TE}$ thus captures the second.
This allows us to represent nonmalleable information flow as a 4-safety
property very simply: $\mathbf{NMIF} = \mathbf{RD} ∩ \mathbf{TE}$.

Figure~\ref{fig:hyper} illustrates the relation between these
hyperproperty definitions.  Observe that the 2-safety hyperproperty
$\mathbf{NI}$ for noninterference is contained in all three
4-safety hyperproperties. The insecure example programs of
\cref{sec:motivation} are found in the left crescent, satisfying $\mathbf{RD}$
but not $\mathbf{NMIF}$.

\begin{figure}
\begin{center}
\includegraphics[height=9em]{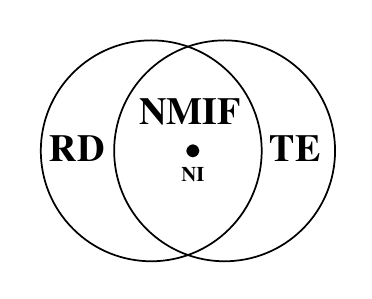}
\end{center}
\vspace{-1em}
\caption{Relating 4-safety hyperproperties and noninterference.}
\label{fig:hyper}
\end{figure}

\section{Implementing NMIF} 
\label{sec:impl}

\lstset{
  frame=none,
  xleftmargin=2pt,
  stepnumber=1,
  numbers=none,
  numbersep=5pt,
  numberstyle=\ttfamily\scriptsize\color[gray]{0.3},
  belowcaptionskip=\bigskipamount,
  captionpos=b,
  escapeinside={(*}{*)},
  language=haskell,
  tabsize=2,
  emphstyle={\bf},
  commentstyle=\it\color[gray]{0.3},
  stringstyle=\mdseries\ttfamily,
  showspaces=false,
  keywordstyle=\bfseries\ttfamily,
  columns=flexible,
  basicstyle=\small\linespread{1}\ttfamily, 
  keepspaces=true,
  showstringspaces=false,
}
\label{sec:enforce}

We have implemented the rules for nonmalleable information flow in
context of Flame, a Haskell library and GHC~\cite{ghc}
plugin.  Flame provides data structures and compile-time checking of
type-level acts-for constraints that are checked using a
custom type-checker plugin.
These constraints are used as the basis for encoding \nmlang
as a shallowly-embedded domain-specific language (DSL).
We have demonstrated that programs
enforcing nonmalleable information flow can be built using this new DSL.

\subsection{Information-flow monads in Flame}

The DSL works by wrapping sensitive information in an abstract data
type---a monad---that includes a principal type parameter representing the
confidentiality and integrity of the information.  

The Flame library tracks computation on protected information as a
monadic effect and provides operations that constrain such
computations to enforce information security. This effect is modeled
using the "IFC" type class defined in Figure~\ref{fig:flameIFC}.
The type class "IFC" is parameterized by two additional types, "n" in the "Labeled"
type class and "e" in "Monad".  Instances of the "Labeled" type class
enforce noninterference on pure computation—no downgrading or effects.
The "e" parameter represents an effect we want to control.  For
instance, many Flame libraries control effects in the "IO" monad,
which is used for input, output, and mutable references.

The type "m e n pc l a" in Figure~\ref{fig:flameIFC} associates a
label "l" with the result of a computation of type "a", as well as a
program counter label "pc" that bounds the confidentiality and
integrity of side effects for some effect "e".  Confidentiality and integrity projections 
are represented by type constructors "C" and "I". The "protect" operator
corresponds to monadic unit $η$ (rule \ruleref{UnitM}).  Given any term,
"protect" labels the term and lifts it into an "IFC" type where "pc ⊑ l".

The "use" operation corresponds to a "bind" term in \nmlang.  Its
constraints implement the \ruleref{BindM} typing rule.  Given a
protected value of type "m e n pc l a" and a function on a value of type "a" with
return type "m e n pc' l' b", "use" returns the result of applying the
function, provided that "l ⊑ l'" and "(pc ⊔ l) ⊑ pc'".  Finally,
"runIFC" executes a protected computation, which results in a labeled
value of type "(n l a)" in the desired effect "e".

\newcommand\findef{\vspace{0.6em}}
\begin{figure}
\begin{lstlisting}
class (Monad e, Labeled n) => IFC m e n where
  protect :: (pc ⊑ l) => a -> m e n pc l a (*\findef*)
  use :: (l ⊑ l', pc ⊑ pc', l ⊑ pc', pc ⊑ pc'') =>
      m e n pc l a -> (a -> m e n pc' l' b)
                   -> m e n pc'' l' b (*\findef*)
  runIFC :: m e n pc l a -> e (n l a)
\end{lstlisting}
\caption{Core information flow control operations in Flame.}
\label{fig:flameIFC}
\end{figure}

We provide "NMIF", which extends the "IFC" type
class with "endorse" and "declassify" operations. The constraints on
 these operations implement the typing rules
\ruleref{Endorse} and \ruleref{Decl} respectively.

\begin{figure}
\begin{lstlisting}
class IFC m e n => NMIF m e n where
  declassify :: ( (C pc) ⊑ (C l)
                , (C l') ⊑ (C l) ⊔ Δ(I (l' ⊔ pc))
                , (I l') === (I l)) =>
             m e n pc l' a -> m e n pc l a (*\findef*) 
     endorse :: ( (I pc) ⊑ (I l)
                , (I l') ⊑ (I l) ⊔ ∇(C (l' ⊔ pc))
                , (C l') === (C l)) =>
             m e n pc l' a -> m e n pc l a 
\end{lstlisting}
\caption{Nonmalleable information flow control in Flame.}
\label{fig:flaops}
\end{figure}

We implemented the secure and insecure sealed-bid auction examples from
Section~\ref{sec:cheat-auction} using "NMIF" operations, shown in
Figure~\ref{fig:flameauc}.
As expected, the insecure "badrecv" is rejected by the compiler while the
secure "recv" type checks.

\begin{figure}
\begin{lstlisting}
recv :: (NMIF m e n, (I p) ⊑ ∇(C p)) =>
     n p a
     -> m e n (I (p ∧ q)) (p ∧ (I q)) a
recv v = endorse $ lift v (*\findef*)
badrecv :: (NMIF m e n, (I p) ⊑ ∇(C p)) =>
        n (p ∧ C q) a
        -> m e n (I (p ∧ q)) (p ∧ q) a
badrecv v = endorse $ lift v (*\textit{\textcolor{red}{\{-REJECTED-\}}}*)
\end{lstlisting}
\caption{Receive operations in \texttt{NMIF}.
  The secure \texttt{recv} is accepted, but the insecure \texttt{badrecv} is
  rejected.}
\label{fig:flameauc}
\end{figure}

\subsection{Nonmalleable HTTP Basic Authentication}

To show the utility of \nmlang, we adapt a simple existing Haskell web
application~\cite{memodb} based on the Servant~\cite{servant} framework to run
in Flame.  
The application allows users to create, fetch, and delete shared memos.
Our version uses HTTP Basic Authentication and Flame's security mechanisms to
restrict access to authorized users.
We have deployed this application online at \demourl.

\begin{figure}
\begin{lstlisting}[basicstyle=\linespread{0.9}\ttfamily\small]
authCheck :: Lbl MemoClient BasicAuthData 
          -> NM IO (I MemoServer) (I MemoServer) 
                   (BasicAuthResult Prin) 
authCheck lauth =
 let lauth' = endorse $ lift lauth
     res = use lauth' $ \(BasicAuthData user guess) ->
           ebind user_db $ \db ->
           case Map.lookup user db of
            Nothing  -> protect Unauthorized
            Just pwd -> 
              if guess == pwd then
               protect $ Authorized (Name user)
              else
               protect Unauthorized 
  in declassify res
\end{lstlisting}
\caption{A nonmalleable password checker in Servant.}
\label{fig:flamebasic}
\end{figure}

Figure~\ref{fig:flamebasic} contains the function "authCheck", which
checks passwords in this application using the "NM" data type, which is an 
instance of the "NMIF" type class.  The function takes a
value containing the username and password guess of the authentication
attempt, labeled with the confidentiality and integrity of an
unauthenticated client, "MemoClient".  This value is endorsed to have
the integrity of the server, "MemoServer".  This operation is safe
since it only endorses information visible to the client.  Next, the
username is used to look up the correct password and compare it to the
client's guess. If they match, then the user is authorized.  The
result of this comparison is secret, so before returning the result,
it must be declassified.

Enforcing any form of information flow control on authentication
mechanisms like "authCheck" provides more information security
guarantees than traditional approaches.
Unlike other approaches, nonmalleable information flow offers strong guarantees
even when a computation endorses untrusted information.
This example shows it is possible to construct applications
that offer these guarantees.

\section{Related work}
\label{sec:related}

Our efforts belong both within a significant body of work attempting
to develop semantic security conditions that are more nuanced than
noninterference, and within an overlapping body of work aiming to
create expressive practical enforcement mechanisms for
information flow control. Most prior work focuses on
relaxing confidentiality restrictions; work
permitting downgrading of integrity imposes
relatively simple controls and lacks
semantic security conditions that capture the concerns exemplified
in \cref{sec:motivation}.

\emph{Intransitive noninterference}~\cite{Rushby92, Pinsky95,
Roscoe99,vanderMeyden2007} is an information flow condition that
permits information to flow only between security levels (or
\emph{domains}) according to some (possibly intransitive) relation.
It does not address the concerns of nonmalleability.

Decentralized information flow control (DIFC)~\cite{ml-tosem} introduces the
idea of mediating downgrading using access control~\cite{pottier00}.
However, the lack of robustness and transparency means downgrading
can still be exploited in these systems
(e.g.,~\cite{myers-popl99,asbestos,histar,flume}).

Robust declassification and qualified robustness have been explored in
DIFC systems as a way to constrain the adversary's influence on
declassification~\cite{zm01b,zznm02,msz06,cm08,am11,flam,flac}.
While transparent endorsement can be viewed as an integrity
counterpart to robust declassification, this idea is not
present in prior work.

Sabelfeld and Sands provide a clarifying taxonomy
for much prior work on declassification~\cite{ss05},
introducing various dimensions along which declassification
mechanisms operate. They categorize robust declassification as lying on
the ``who'' dimension. However, they do not explicitly consider
endorsement mechanisms. Regardless of the taxonomic category,
transparent endorsement and nonmalleable information flow also seem to lie on
the same dimension as robust declassification, since they take into account
influences on the information that is downgraded.

Label algebras~\cite{lblalgebra} provide an abstract characterization
of several DIFC systems. However, they do not address the restrictions
on downgrading imposed by nonmalleable information flow.

The Aura language~\cite{aura} uses information flow policies to
constrain authorization and endorsement. However, it does not address
the malleability of endorsement. Rx~\cite{shtz06} represents
information flow control policies in terms of dynamic
_roles_~\cite{RBAC}.  Adding new principals to these roles corresponds
to declassification and endorsement since new flows may occur.  Rx
constrains updates to roles similarly to previous
type systems that enforce robust declassification and qualified
robustness but does not prevent opaque endorsements.

Relational Hoare Type Theory~\cite{rhtt} (RHTT) offers a powerful and precise
way to specify security conditions that are 2-hyperproperties,
such as noninterference.
Cartesian Hoare logic~\cite{sousa2016cartesian} (CHL) extends standard Hoare logic to
reason about $k$-safety properties of relational traces (the input/output pairs of a program).
Since nonmalleable information flow, robust declassification, and transparent endorsement
are all 4-safety properties that cannot be fully expressed with relational traces, neither RHTT nor CHL can 
characterize them properly.

Haskell's type system has been attractive
target for embedding information flow checking~\cite{lz06,
stefan:2014:building-haskell, hlio}.  Much prior work has focused on
dynamic information flow control.
LIO~\cite{stefan:2014:building-haskell}
requires computation on protected information to occur
in the "LIO" monad, which tracks the
confidentiality and integrity of information accessed (``unlabeled'') by the computation. HLIO~\cite{hlio} explores hybrid
static and dynamic enforcement.  Flame enforces
information flow control statically, and the "NMIF" type class
enforces nonmalleable IFC statically as well.
The static component of HLIO enforces solely via the Haskell
type system (and existing general-purpose extensions), but Flame---and by
extension, "NMIF"---uses custom constraints based on the FLAM algebra
which are processed by a GHC type checker plugin. Extending the
type checker to reason about FLAM constraints significantly improves
programmability over pure-Haskell approaches like HLIO.

\section{Conclusion}
\label{sec:conclusions}

Downgrading mechanisms like declassification and endorsement 
make information flow mechanisms sufficiently flexible and
expressive for real programs. However, we have shown that previous
notions of information-flow security missed the dangers endorsing
confidential information.
We therefore
introduced transparent endorsement as a security property that rules
out such influences and showed that it is dual to robust
declassification. Robust declassification and transparent endorsement
are both consequences of a stronger property, nonmalleable information
flow, and we have formulated all three as 4-safety properties.
We have shown how to provably enforce these security
properties in an efficient, compositional way using a security type system.
 Based on our Haskell implementation, these security conditions and
enforcement mechanism appear to be practical, supporting the secure
construction of programs with complex security requirements.

While security-typed languages are not yet mainstream,
information flow control, sometimes in the guise of taint tracking,
has become popular as a way to detect and control
real-world vulnerabilities~(e.g., \cite{ifappstore}). Just as the
program analyses used are approximations of previous security type
systems targeting noninterference, it is reasonable to expect the
\nmlang type system to be a useful guide for other analyses and enforcement
mechanisms.

\section*{Acknowledgments}

Many people helped us with this work.  Martín Abadi posed a
provocative question about dualities.  Pablo Buiras helped develop the
memo example. David Naumann pointed out work on $k$-safety.  Tom
Magrino, Yizhou Zhang, and the anonymous reviewers gave us useful
feedback on the paper.

Funding for this work was provided by NSF grants 1513797 and 1524052,
and by a gift from Google.
Any opinions, findings, conclusions, or recommendations expressed
here are those of the author(s) and do not necessarily reflect those
of these sponsors.

\bibliographystyle{ACM-Reference-Format}
\bibliography{../bibtex/pm-master}

\appendices

\section{Full \nmlang}
\label{sec:full-nmlang}

We present the full syntax, semantics, and typing rules for \nmlang in
Figures~\ref{fig:full-syntax}, \ref{fig:full-semantics}, and
\ref{fig:full-types}, respectively.
This is a straightforward extension of the core language presented in
Section~\ref{sec:nmlang}.
We note that polymorphic terms specify a $\pc$ just as $λ$ terms.
This is because they contain arbitrary expressions which could produce
arbitrary effects, so we must constrain the context that can execute those
effects.

\begin{figure}
  \small
  \[
    \begin{array}{rcl}
      \multicolumn{3}{l}{ n ∈ \N \text{  (atomic principals)}} \\
      \multicolumn{3}{l}{ x ∈ \mathcal{V} \text{  (variable names)}} \\
      \\
      p,ℓ,\pc &::=&  n \sep \top \sep \bot \sep p^{π} \sep p ∧ p\sep p ∨ p \sep p ⊔ p \sep p ⊓ p \\[0.8em]
      τ &::=& \voidtype \sep X \sep \sumtype{τ}{τ} \sep \prodtype{τ}{τ} \\[0.4em]
        & \sep & \func{τ}{\pc}{τ} \sep \tfunc{X}{\pc}{τ} \sep \says{ℓ}{τ} \\[0.8em]
      v &::=& \void \sep \inji{v} \sep \pair{v}{v} \sep \vreturn{ℓ}{v} \\[0.4em]
        & \sep & \lamc{x}{τ}{\pc}{e} \sep \tlam{X}{\pc}{e} \\[0.8em]
      e &::=& x \sep v \sep e~e \sep e~τ \sep \pair{e}{e} \sep \return{ℓ}{e} \\[0.4em]
        & \sep & \proji{e} \sep \inji{e} \sep \bind{x}{e}{e} \\[0.4em]
        & \sep & \casexp{e}{x}{e}{e} \\[0.4em]
        & \sep & \cdowngrade{e}{ℓ} \sep \idowngrade{e}{ℓ}
    \end{array}
  \]
  \caption{Full \nmlang syntax.}
  \label{fig:full-syntax}
\end{figure}

\begin{figure}
\begin{flushleft}
  \small
  \boxed{e \stepsone e'} \\
  \begin{mathpar}
    \erule{E-App}{}{(\lamc{x}{τ}{\pc}{e})~v}{\subst{e}{x}{v}}

    \erule{E-TApp}{}{(\tlam{X}{\pc}{e})~τ}{\subst{e}{X}{τ}}

    \erule{E-UnPair}{}{\proji{\pair{v_1}{v_2}}}{v_i}

    \erule{E-Case}{}{(\casexp{(\inji{v})}{x}{e_1}{e_2})}{\subst{e_i}{x}{v}}

    \erule{E-BindM}{}{\bind{x}{\vreturn{ℓ}{v}}{e}}{\subst{e}{x}{v}}
    \hfill
  \end{mathpar}

  \boxed{\evalctx*{e} \stepsonectx \evalctx'{e'}} \\
  \begin{mathpar}
    \erulectx{E-Step}{e \stepsone e'}{\evalctx*{e}}{\evalctx{e'}{t;\bullet}}

    \erulectx{E-UnitM}{}{\evalctx*{\return{ℓ}{v}}}{\evalctx{\vreturn{ℓ}{v}}{t;\vreturn{ℓ}{v}}}

    \erulectx{E-Decl}{}{\evalctx*{\cdowngrade{\vreturn{ℓ'}{v}}{ℓ}}}{\evalctx{\vreturn{ℓ}{v}}{t;\downemit*{→}{v}}}

    \erulectx{E-Endorse}{}{\evalctx*{\idowngrade{\vreturn{ℓ'}{v}}{ℓ}}}{\evalctx{\vreturn{ℓ}{v}}{t;\downemit*{←}{v}}}

    \erulectx{E-Eval}{\evalctx*{e} \stepsonectx \evalctx'{e'}}{\evalctx*{E[e]}}{\evalctx'{E[e']}}
    \hfill
  \end{mathpar}

  \vspace{1em}
  \underline{Evaluation context}
  \[
    \begin{array}{rcl}
      E & ::= & [\cdot] \sep E~e \sep v~E \sep E~τ \sep \pair{E}{e} \sep \pair{v}{E} \sep \return{ℓ}{E} \\[0.4em]
        & \sep & \proji{E} \sep \inji{E} \sep \bind{x}{E}{e}  \\[0.4em]
        & \sep & \casexp{E}{x}{e}{e} \\[0.4em]
        & \sep & \cdowngrade{E}{ℓ} \sep \idowngrade{E}{ℓ} \\[0.8em]
    \end{array}
  \]
\end{flushleft}
\caption{Full \nmlang operational semantics.}
\label{fig:full-semantics}
\end{figure}

\begin{figure}
\begin{flushleft}
  \rulefiguresize
  \boxed{\protjudge*{ℓ}{τ}}
  \begin{mathpar}
    \protrule{P-Unit}{}{ℓ}{\voidtype}

    \protrule{P-Lbl}{\stflowjudge*{ℓ'}{ℓ}}{ℓ'}{\says{ℓ}{τ}}

    \protrule{P-Pair}{
      \protjudge*{ℓ}{τ_1} \\
      \protjudge*{ℓ}{τ_2}
    }{ℓ}{\prodtype{τ_1}{τ_2}}
    \hfill
  \end{mathpar}

  \boxed{\protjudge*{τ}{\H}}
  \begin{mathpar}
    \Rule[$\H$ is upward closed]{P-Set}{%
      H ∈ \H \\
      \protjudge*{H}{τ}
    }{\hightype{\H}{τ}}
    \hfill
  \end{mathpar}
\end{flushleft}
\caption{Type protection levels.}
\label{fig:full-protect}
\end{figure}

\begin{figure}
\begin{flushleft}
  \rulefiguresize
  \boxed{\TValGpc{e}{τ}} \\
  \begin{mathpar}
    \Rule{Var}{}{\TVal{Γ,x\ty τ,Γ';\pc}{x}{τ}}

    \Rule{Unit}{}{\TValGpc{\void}{\voidtype}}

    \Rule{Lam}{%
      \TVal{Γ,x\ty τ_1;\pc'}{e}{τ_2}
    }{\TValGpc{\lamc{x}{τ_1}{\pc'}{e}}{\func{τ_1}{\pc'}{τ_2}}}

    \Rule{App}{%
      \TValGpc{e_1}{\func{τ'}{\pc'}{τ}} \\\\
      \TValGpc{e_2}{τ'} \\
      \stflowjudge*{\pc}{\pc'}
    }{\TValGpc{e_1~e_2}{τ}}

    \Rule{TLam}{%
      \TVal{Γ,X;\pc'}{e}{τ}
    }{\TValGpc{\tlam{X}{\pc'}{e}}{\tfunc{X}{\pc'}{τ}}}

    \Rule[$τ'$ is well-formed in $Γ$]{%
      TApp}{\TValGpc{e}{\tfunc{X}{\pc'}{τ}} \\\\
      \stflowjudge*{\pc}{\pc'}
    }{\TValGpc{(e~τ')}{\subst{τ}{X}{τ'}}}

    \Rule{Pair}{%
      \TValGpc{e_1}{τ_1} \\
      \TValGpc{e_2}{τ_2}
    }{\TValGpc{\pair{e_1}{e_2}}{\prodtype{τ_1}{τ_2}}}

    \Rule{UnPair}{%
      \TValGpc{e}{\prodtype{τ_1}{τ_2}}
    }{\TValGpc{\proji{e}}{τ_i}}

    \Rule{Inj}{%
      \TValGpc{e}{τ_i}
    }{\TValGpc{\inji{e}}{\sumtype{τ_1}{τ_2}}}

    \Rule{Case}{%
      \TValGpc{e}{\sumtype{τ_1}{τ_2}} \\
      \protjudge*{\pc}{τ} \\\\
      \TVal{Γ,x\ty τ_1;\pc}{e_1}{τ} \\
      \TVal{Γ,x\ty τ_2;\pc}{e_2}{τ} \\
    }{\TValGpc{\casexp{e}{x}{e_1}{e_2}}{τ}}

    \Rule{UnitM}{%
      \TValGpc{e}{τ} \\
      \stflowjudge*{\pc}{ℓ}
    }{\TValGpc{\return{ℓ}{e}}{\says{ℓ}{τ}}}

    \Rule{VUnitM}{
      \TValGpc{v}{τ}
    }{\TValGpc{\vreturn{ℓ}{v}}{\says{ℓ}{τ}}}

    \Rule{BindM}{%
      \TValGpc{e}{\says{ℓ}{τ'}} \\
      \protjudge*{ℓ}{τ} \\\\
      \TVal{Γ,x\ty τ';\pc ⊔ ℓ}{e'}{τ} \\
    }{\TValGpc{\bind{x}{e}{e'}}{τ}}

    \Rule{Decl}{%
      \TValGpc{e}{\says{ℓ'}{τ}} \\
      ℓ'^{←} = ℓ^{←} \\
      \stflowjudge*{\pc}{ℓ} \\\\
      \stflowjudge*{ℓ'^{→}}{ℓ^{→} ⊔ \view{(ℓ' ⊔ \pc)^{←}}} \\
    }{\TValGpc{\cdowngrade{e}{ℓ}}{\says{ℓ}{τ}}}

    \Rule{Endorse}{%
      \TValGpc{e}{\says{ℓ'}{τ}} \\
      ℓ'^{→} = ℓ^{→} \\
      \stflowjudge*{\pc}{ℓ} \\\\
      \stflowjudge*{ℓ'^{←}}{ℓ^{←} ⊔ \voice{(ℓ' ⊔ \pc)^{→}}} \\
    }{\TValGpc{\idowngrade{e}{ℓ}}{\says{ℓ}{τ}}}
    \hfill
  \end{mathpar}
\end{flushleft}
\caption{Typing rules for full \nmlang language.}
\label{fig:full-types}
\end{figure}

Figure~\ref{fig:af-rules} presents the full set of derivation rules for the
acts-for (delegation) relation $p ≽ q$.

\begin{figure}
\begin{flushleft}
  \rulefiguresize
  \boxed{\stafjudge*{p}{q}} \\
  \begin{mathpar}
    \Rule{Bot}{}{\stafjudge*{p}{⊥}}
    \and
    \Rule{Top}{}{\stafjudge*{⊤}{p}}
    \and
    \Rule{Refl}{}{\stafjudge*{p}{p}}
    \and
    \Rule{Proj}{\stafjudge*{p}{q}}{\stafjudge*{p^{π}}{q^{π}}}
    \and
    \Rule{ProjR}{}{\stafjudge*{p}{p^{π}}}
    \and
    \Rule{ConjL}{
      \stafjudge*{p_i}{q} \\\\
      i ∈ \{1,2\}
    }{\stafjudge*{p_1 ∧ p_2}{q}}
    \and
    \Rule{ConjR}{
      \stafjudge*{p}{q_1} \\\\
      \stafjudge*{p}{q_2}
    }{\stafjudge*{p}{q_1 ∧ q_2}}
    \and
    \Rule{DisL}{
      \stafjudge*{p_1}{q} \\\\
      \stafjudge*{p_2}{q}
    }{\stafjudge*{p_1 ∨ p_2}{q}}
    \and
    \Rule{DisR}{
      \stafjudge*{p}{q_i} \\\\
      i ∈ \{1,2\}
    }{\stafjudge*{p}{q_1 ∨ q_2}}
    \and
    \Rule{Trans}{
      \stafjudge*{p}{q} \\
      \stafjudge*{q}{r}
    }{\stafjudge*{p}{r}}
    \hfill
  \end{mathpar}
\end{flushleft}
\caption{Principal lattice rules}
\label{fig:af-rules}
\end{figure}

\subsection{Label tracking with brackets}

In order to simply proofs of hyperproperties requiring 2 and 4 traces, we
introduce a new bracket syntax to track secret and untrusted data.
These brackets are inspired by those used by Pottier and Simonet~\cite{ps03} to
prove their FlowCaml type system enforced noninterference.
Their brackets served two purposes simultaneously.
First they allow a single execution of a bracketed program to faithfully model
two executions of a non-bracketed program.
Second, the brackets track secret/untrusted information through execution of
the program, thereby making it easy to verify that it did not interfere with
public/trusted information simply by proving that brackets could not be
syntactically present in such values.
Since noninterference only requires examining pairs of traces, these purposes
complement each other well; if the two executions vary only on high inputs,
then low outputs cannot contain brackets.

While this technique is very effective to prove noninterference, nonmalleable
information flow provides security guarantees even in the presence of both
declassification and endorsement.
As a result, we need to track secret/untrusted information even through
downgrading events that can cause traces to differ arbitrarily.
To accomplish this goal, we use brackets that serve only the second purpose:
they track restricted information but not multiple executions.

As in previous formalizations, \nmlang's brackets are defined with respect to a
notion of ``high'' labels, in this case a high set.
The high set restricts the type of the expression inside the bracket as well as
the $\pc$ at which it must type, thereby restricting the effects it can create.
For the more complex theorems we must track data with multiple different high
labels within the same program execution, so we parameterize the brackets
themselves with the high set.
We present the extended syntax, semantics, and typing rules in Figure~\ref{fig:extensions}.

\begin{figure}
\begin{flushleft}
  \small
  \underline{Syntax extensions}
  \begin{align*}
    v \quad & ::= \quad \cdots \sep \lab*{v} \\
    e \quad & ::= \quad \cdots \sep \lab*{e}
  \end{align*}

  \underline{New contexts}
  \begin{align*}
    E \quad & ::= \quad \cdots \sep \lab*{E} \\
    \bctx \quad & ::= \quad \proji{[\cdot]} \sep \bind{x}{[\cdot]}{e}
  \end{align*}

  \underline{Evaluation extensions}
  \begin{mathpar}
    \erule{B-Expand}{}{\bctx[\lab*{v}]}{\lab*{\bctx[v]}}

    \erule{B-DeclL}{ℓ \notin \H}{\cdowngrade{\lab*{v}}{ℓ}}{\cdowngrade{v}{ℓ}}

    \erule{B-DeclH}{ℓ ∈ \H}{\cdowngrade{\lab*{v}}{ℓ}}{\lab*{\cdowngrade{v}{ℓ}}}

    \erule{B-EndorseL}{ℓ \notin \H}{\idowngrade{\lab*{v}}{ℓ}}{\idowngrade{v}{ℓ}}

    \erule{B-EndorseH}{ℓ ∈ \H}{\idowngrade{\lab*{v}}{ℓ}}{\lab*{\idowngrade{v}{ℓ}}}
    \hfill
  \end{mathpar}

  \underline{Typing extensions}
  \begin{mathpar}
    \Rule[$\H$ is upward closed]{Bracket}{%
      \TVal{Γ;\pc'}{e}{τ} \\
      \stflowjudge*{\pc}{\pc'} \\\\
      \pc' ∈ \H \\
      \hightype{\H}{τ}
    }{\TValGpc{\lab*{e}}{τ}}
    \hfill
  \end{mathpar}

  \underline{Bracket projection}
  \[
    \bget{}{e} = \begin{cases}
      \bget{}{e'} \quad \text{if } e = \lab*{e'} \\
      \text{recursively project all sub-expressions otherwise}
    \end{cases}
    \]
\end{flushleft}
\caption{\nmlang language extensions.}
\label{fig:extensions}
\end{figure}

\section{Attacker properties}
\label{sec:attacker-properties}

Recall that we defined an attacker as a set of principals
$\A = \{ ℓ ∈ \L \mid \stafjudge*{n_1 ∧ \dotsb ∧ n_k}{ℓ} \}$
for some non-empty finite set of atomic principals $\{n_1, \dotsc, n_k\} ⊆ \N$.

\begin{definition}[Attacker properties]
  \label{def:attacker}
  Let $\A$ be an attacker and let $\A^{π} = \{ p ∈ \L \mid \exists q ∈ \L \text{ such that } p^{π} ∧ q^{π'} ∈ \A \}$.
  The following properties hold:
  \begin{enumerate}
    \item \label{li:atk:a-closed-conj} for all $a_1, a_2 ∈ \A^{π}$, $a_1 ∧ a_2 ∈ \A^{π}$
          (Attacking principals may collude)
    \item \label{li:atk:disj-with-a} for all $a ∈ \A^{π}$ and $b ∈ \L$, $a ∨ b ∈ \A^{π}$
          (Attackers may attenuate their power)
    \item \label{li:atk:not-a-closed-disj} for all $b_1, b_2 \notin \A^{π}$, $b_1 ∨ b_2 \notin \A^{π}$
          (Combining public information in a public context yields public information and
           combining trusted information in a trusted context yields trusted information)
    \item \label{li:atk:conj-with-not-a} for all $a ∈ \L$ and $b \notin \A^{π}$, $a ∧ b \notin \A^{π}$
          (Attackers cannot compromise policies requiring permission from non-attacking principals)
    \item \label{li:atk:voice-and-view} for all $a ∈ \A$, $\voice{a^{→}} ∧ \view{a^{←}} ∈ \A$.
          (Attackers have the same power in confidentiality and integrity)
  \end{enumerate}
\end{definition}

The theorems proved in this paper hold for any attacker satisfying
these properties, so for generality we can take the properties
in Definition~\ref{def:attacker} as defining an attacker.

We now prove that our original definition of an attacker satisfies
Definition~\ref{def:attacker}.

\begin{proof}
  Conditions~\ref{li:atk:a-closed-conj} and \ref{li:atk:disj-with-a} of
  Definition~\ref{def:attacker} follow directly from the definition of $\A$ and
  \ruleref{ConjR} and \ruleref{DisR}, respectively.
  Condition~\ref{li:atk:voice-and-view} holds by the symmetry of the lattice.

  Since we are only examining one of confidentiality and integrity at a time,
  for the following conditions we assume without loss of generality that all
  principals in each expression have only the $π$ projection and the other
  component is $⊥$.
  In particular, this means we can assume \ruleref{Proj} and \ruleref{ProjR}
  are not used in any derivation, and any application of the conjunction or
  disjunction derivation rules split in a meaningful way with respect to the
  $π$ projection (i.e., neither principal in the side being divided is $⊤$ or
  $⊥$).

  We now show Condition~\ref{li:atk:conj-with-not-a} holds by contradiction.
  Assume $a ∈ \L$ and $b \notin \A^{π}$, but $a ∧ b ∈ \A^{π}$.
  This means $\stafjudge*{(n_1 ∧ \dotsb ∧ n_k)^{π}}{a ∧ b}$.
  We prove by induction on $k$ that $a, b ∈ \A^{π}$.
  If $k = 1$, then the only possible rule to derive this result is
  $\ruleref{ConjL}$ and we are finished.
  If $k > 1$, then the derivation of this relation must be due to either
  \ruleref{ConjL} or \ruleref{ConjR}.
  If it is due to \ruleref{ConjR}, then this again achieves the desired
  contradiction.
  If it is due to \ruleref{ConjL}, then the same statement holds for a subset
  of the atomic principals $n'_1, \dotsc, n'_{k'}$, where $k' < k$, so by
  induction, $\stafjudge*{(n'_1 ∧ \dotsb ∧ n'_{k'})^{π}}{b^{π}}$, and by
  \ruleref{Trans}, $\stafjudge*{(n_1 ∧ \dotsc ∧ n_k)^{π}}{b^{π}}$ which also
  contradicts our assumption.

  Finally, we also show Condition~\ref{li:atk:not-a-closed-disj} holds by
  contradiction.
  We assume $b_1, b_2 \notin \A^{π}$ but $b_1 ∨ b_2 ∈ \A^{π}$ and again prove a
  contradiction by induction on $k$.
  If $k = 1$, then the derivation showing $\stafjudge*{n_1^{π}}{(b_1∨b_2)^{π}}$
  must end with \ruleref{DisR} which contradicts the assumption that
  $b_1, b_2 \notin \A^{π}$.
  If $k > 1$, the derivation either ends with \ruleref{DisR}, resulting in the
  same contradiction, or with \ruleref{ConjL}.
  In this second case, the same argument as above holds: there is a strict
  subset of the principals $n_1, \dotsc, n_k$ that act for either $b_1$ or
  $b_2$ and thus by \ruleref{Trans} we acquire the desired contradiction.
\end{proof}

\section{Generalization}
\label{sec:nmif-generalize}

Definition~\ref{def:nmifc} (and correspondingly Theorem~\ref{thm:nmifc}) might
appear relatively narrow;
they only speak directly to programs with a single untrusted value and a single secret
value. However, because the language has first-class functions and pair types, the theorem
as stated is equivalent to one that allows insertion of secret and untrusted
code into multiple points in the program, as long as that code types in an appropriately restrictive $\pc$.

To define this formally, we first need a means to allow for insertion of
arbitrary code.
We follow previous work~\cite{msz06} by extending the
language to include _holes_.
A program expression may contain an ordered set of holes. These holes
may be replaced with arbitrary expressions, under restrictions requiring that the
holes be treated as sufficiently secret or untrusted.
Specifically, the type system is extended with the following rule:
\begin{center}
  \small
  \begin{mathpar}
    \Rule[$\H$ is a high set]{Hole}{%
      \pc ∈ \H \\
      \hightype{\H}{τ}
    }{\TValGpc{[\bullet]_{\H}}{τ}}
    \hfill
  \end{mathpar}
\end{center}

Using this definition, we can state NMIF in a more traditional form.

\begin{definition}[General NMIF]
  \label{def:general-nmif}
  We say that a program $e[\vec{\bullet}]_{\H}$ enforces _general NMIF_ if the
  following holds for all attackers $\A$ inducing high sets $\U$ and $\secret$.
  Let $\trusted = \L \setminus \U$, $\public = \L \setminus \secret$ and
  $\low = \trusted ∩ \secret$.
  If $\H ⊆ \U$, then for all values $v_1$, $v_2$ and all attacks $\vec{a}_1$
  and $\vec{a}_2$, let
  \[ \evalctx{\subst{e[\vec{a}_i]_{\H}}{\vec{x}}{\vec{v}_i}}{\vec{v}_i} \stepstoctx \evalctx{v_{ij}}{t^{ij}}. \]
  For all indices $n_{ij}$ such that $\nseteq{\low}{\trelt{t^{ij}}{n_{ij}}}{\bullet}$
  \begin{enumerate}[leftmargin=*]
    \item \label{li:nmifc:rd}
      If $\traceeq{\trusted}{\trpref{t^{i1}}{n_{i1} - 1}}{\trpref{t^{i2}}{n_{i2} - 1}}$ for
      $i = 1,2$, then
      \[ \left(\rel*{←}{w_1} \text{ and } \traceeq{\public}{\trpref{t^{11}}{n_{11}}}{\trpref{t^{21}}{n_{21}}}\right) ~ \Longrightarrow ~ \traceeq{\public}{\trpref{t^{12}}{n_{12}}}{\trpref{t^{22}}{n_{22}}}. \]
    \item \label{li:nmifc:te}
      Similarly, if $\traceeq{\public}{\trpref{t^{1j}}{n_{1j} - 1}}{\trpref{t^{2j}}{n_{2j} - 1}}$
      for $j = 1,2$, then
      \[ \left(\rel*{→}{v_1} \text{ and } \traceeq{\trusted}{\trpref{t^{11}}{n_{11}}}{\trpref{t^{12}}{n_{12}}}\right) ~ \Longrightarrow ~ \traceeq{\trusted}{\trpref{t^{21}}{n_{21}}}{\trpref{t^{22}}{n_{22}}}. \]
  \end{enumerate}
\end{definition}

For \nmlang, this definition is equivalent to Definition~\ref{def:nmifc}.
We prove this fact to prove the following theorem.

\begin{theorem}[General NMIF]
  \label{thm:general-nmif}
  Given a program $e[\vec{\bullet}]_{\H}$ such that
  $\TVal{Γ,\vec{x}\ty \vec{τ};\pc}{e[\vec{\bullet}]_{\H}}{τ'}$,
  then $e[\vec{\bullet}]_{\H}$ enforces general NMIF.
\end{theorem}

\begin{proof}
  We prove this by reducing Definition~\ref{def:general-nmif} to
  Definition~\ref{def:nmifc} in two steps.
  We assume that no two variables in the original expression
  $e[\vec{\bullet}]_{\H}$ have the same name as this can be enforced by
  $α$-renaming.

  The first step handles expressions that only substitute values (and have no
  holes), but allow any number of both secret and untrusted values.
  An expression of the form in this corollary is easily rewritten
  as such a substitution as follows.
  For each hole $[\bullet]_{\H}$, we note that
  $\TVal{Γ';\pc'}{[\bullet]_{\H}}{τ''}$ where $Γ,\vec{x}\ty\vec{τ} ⊆ Γ'$
  and $\pc' ∈ \H$.
  We replace the hole with a function application inside a "bind".
  Specifically, the hole becomes
  \[ \bind{y'}{y}{\left(y'~z_1~\dotsb~z_k\right)} \]
  where $y$ and $y'$ are fresh variables and the $z_i$s are every variable in
  $Γ' \setminus Γ$ (including every element of $\vec{x}$).
  Let
  \[ τ_y = \says{\pc'}{\left(τ_{z_1} \xrightarrow{\pc'} \dotsb \xrightarrow{\pc'} τ_{z_k} \xrightarrow{\pc'} τ''\right)} \]
  and include $y\ty τ_y$ as the type of an untrusted value to substitute in.

  Instead of inserting the expression $a$ into that hole, we substitute in for $y$ the value
  \[ w = \vreturnp{\pc'}{\left(\lamc{z_1}{τ_{z_1}}{\pc'}~\dotsb~\lamc{z_k}{τ_{z_k}}{\pc'}{a}\right)}. \]
  By \ruleref{Hole} we know that $\pc' ∈ \H$ and $\hightype{\H}{τ''}$, so
  the type has the proper protection, and by construction $\TValGpc{w}{τ_y}$.
  Moreover, while it has an extra value at the beginning of the trace (the
  function), the rest of the traces are necessarily the same.

  As a second step, we reduce the rest of the way to the expressions used in
  Definition~\ref{def:nmifc}.
  To get from our intermediate step to these single-value expressions, if we
  wish to substitute $k_s$ secret values and $k_u$ untrusted values, we instead
  substitute a single list of $k_s$ secret values and a single
  list of $k_u$ untrusted values.
  These lists are constructed in the usual way out of pairs, meaning the
  protection relations continue to hold as required.
  Finally, whenever a variable is referenced in the unsubstituted expression,
  we instead select the appropriate element out of the substituted list using
  nested projections.
\end{proof}

We also note that the same result holds if we allow for insertion of secret
code and untrusted values, as the argument is exactly dual.
Such a situation, however, makes less sense, so we do not present it explicitly.

\section{Proofs}
\label{sec:proofs}

We now prove a variety of properties about \nmlang.

\subsection{Language results}

We begin with the core results about the language itself.

\begin{lemma}[Values]
  \label{lem:values}
  For any value $v$ such that $\TValGpc{v}{τ}$, $\TValG{\pc'}{v}{τ}$ for any $\pc'$.
\end{lemma}

\begin{proof}
  This follows by induction on the typing derivation for values.
\end{proof}

\begin{lemma}[Substitution]
  \label{lem:subst}
  If $\TVal{Γ,x\ty τ';\pc}{e}{τ}$ and $\TValGpc{v}{τ'}$, then $\TValGpc{\subst{e}{x}{v}}{τ}$.
\end{lemma}

\begin{proof}
  By induction on the derivation of $\TVal{Γ,x\ty τ';\pc}{e}{τ}$ using Lemma~\ref{lem:values}.
\end{proof}

\begin{lemma}[$\pc$ reduction]
  \label{lem:pcred}
  If $\TValGpc{e}{τ}$ and $\stflowjudge*{\pc'}{\pc}$, then $\TValG{\pc'}{e}{τ}$.
\end{lemma}

\begin{proof}
  By induction on the derivation of $\TValGpc{e}{τ}$.
\end{proof}

\begin{theorem}[Subject reduction]
  \label{thm:subred}
  If $\TValGpc{e}{τ}$ and $\evalctx*{e} \stepsonectx \evalctx'{e'}$ then $\TValGpc{e'}{τ}$.
\end{theorem}

\begin{proof}
  This proof follows by an inductive case analysis on the operational semantics
  in Figures~\ref{fig:full-semantics} and \ref{fig:extensions}.
  There are a few interesting cases.

  \begin{case}[\ruleref{B-Expand}]
    This case handles a context ($\bctx$), we will do a sub-case analysis
    on each such expression type.
    \begin{itemize}[leftmargin=*]
      \item $e = \left(\proji{\lab*{v}}\right)$:
        \ruleref{UnPair} allows the expression to type-check in any $\pc$ in
        which its argument type-checks.
        Since \ruleref{Bracket} says $\TValG{\pc'}{v}{\prodtype{τ_1}{τ_2}}$ for
        some $\pc' ∈ \H$, we get that $\TValG{\pc'}{\proji{v}}{τ_i}$.
        Moreover \ruleref{Bracket} and \ruleref{P-Set} also require
        $\protjudge*{H}{\prodtype{τ_1}{τ_2}}$ which, by \ruleref{P-Pair} can
        only happen if $\protjudge*{H}{τ_i}$ for $i = 1,2$.
        Together these satisfy the requirements for \ruleref{Bracket} and give
        us $\TValG{\pc'}{\lab*{\proji{v}}}{τ_i}$.

      \item $e = \left(\bind{x}{\lab*{v}}{e'}\right)$:
        By \ruleref{BindM}, $\TValGpc{\lab*{v}}{\says{ℓ}{τ'}}$ and by
        \ruleref{Bracket} and inversion on the protection rules,
        $\stflowjudge*{H}{ℓ}$ for some $H ∈ \H$ and thus $ℓ ∈ \H$.
        Let $\pc' = \pc'' ⊓ (\pc ⊔ ℓ)$ where $\pc''$ is the higher $\pc$ used
        in \ruleref{Bracket}.
        Since $\H$ is upward closed and $\pc'', ℓ ∈ \H$, $(\pc ⊔ ℓ) ∈ \H$ and
        thus $\pc' ∈ \H$.
        Moreover, $\steqjudge*{\pc' ⊔ ℓ}{\pc ⊔ ℓ}$, so by \ruleref{BindM}
        $\TVal{Γ,x\ty τ';\pc' ⊔ ℓ}{e'}{τ}$ and by Lemma~\ref{lem:values}
        $\TValG{\pc'}{v}{\says{ℓ}{τ'}}$.
        This means that $\TValG{\pc'}{\bind{x}{v}{e'}}{τ}$.
        Also by \ruleref{BindM} we have $\protjudge*{ℓ}{τ}$ so since $ℓ ∈
        \H$ we have now satisfied the conditions on \ruleref{Bracket} and
        $\TValGpc{\lab*{\bind{x}{v}{e'}}}{τ}$.
    \end{itemize}
  \end{case}

  \begin{case}[\ruleref{B-DeclH}]
    By \ruleref{Decl} $\TValGpc{\lab*{v}}{\says{ℓ'}{τ'}}$ and by inspection on
    the protection rules and \ruleref{Bracket}, it must be the case that
    $\stflowjudge*{H}{ℓ'}$ for some $H ∈ \H$ and thus $ℓ' ∈ \H$.

    By assumption $ℓ ∈ \H$, so the final protection requirement of
    \ruleref{Bracket} on the stepped expression is satisfied.
    We now claim that if $\pc' = \pc ⊔ ℓ$ then
    $\TValG{\pc'}{\cdowngrade{v}{ℓ}}{τ}$.
    \ruleref{Decl} gives us that $\stflowjudge*{\pc}{ℓ}$, so clearly
    $\stflowjudge*{\pc' = \pc ⊔ ℓ}{ℓ}$.

    \ruleref{Decl} also gives us $\stflowjudge*{ℓ'^{→}}{ℓ^{→} ⊔ \view{(ℓ' ⊔ \pc)^{←}}}$.
    Since $ℓ'^{←} = ℓ^{←}$, we note that
    \[ (ℓ' ⊔ \pc)^{←} = (ℓ ⊔ \pc)^{←} = \pc'^{←}. \]
    Therefore $(ℓ' ⊔ \pc')^{←} = (ℓ' ⊔ \pc)^{←}$ and the premise still holds.

    Since the other premises \ruleref{Decl} are trivially still satisfied in
    a higher $\pc$ (using Lemma~\ref{lem:values} for $v$ to type), we see that
    $\TValG{\pc'}{\cdowngrade{v}{ℓ}}{τ}$.
    Thus since $\pc' ∈ \H$ and $ℓ ∈ \H$, we get
    $\TValGpc{\lab*{\cdowngrade{v}{ℓ}}}{τ}$.
  \end{case}

  \begin{case}[\ruleref{B-EndorseH}]
    We omit the details of this case as they are identical to the previous
    case, but with projection arrows reversed and $\view*$ replaced by $\voice*$.
  \end{case}

  All other cases follow trivially from inspection on the typing derivations in
  Figure~\ref{fig:core-types} and applications of Lemmas~\ref{lem:subst} and
  \ref{lem:pcred}.
  Rule \ruleref{E-Eval} also requires an inductive application.
\end{proof}

\begin{lemma}[Bracket Soundness]
  \label{lem:sound}
  If $\evalctx{e}{ϵ} \stepstoctx \evalctx*{e'}$ then
  $\evalctx{\bget{}{e}}{ϵ} \stepstoctx \evalctx{\bget{}{e'}}{\bget{}{t}}$.
\end{lemma}

\begin{proof}
  This result follows by inspection on the operational semantics in
  Figures~\ref{fig:full-semantics} and \ref{fig:extensions}.
  Note that every step the non-projected expression takes is mirrored exactly
  by a step the projected value takes except applications of
  \ruleref{B-Expand}, \ruleref{B-DeclL}, \ruleref{B-DeclH},
  \ruleref{B-EndorseL}, and \ruleref{B-EndorseH}.
  Those applications are simply dropped and since they emit no values to the
  trace, the traces remain the same.
\end{proof}

\begin{lemma}[Bracket Completeness]
  \label{lem:complete}
  If $\evalctx{\bget{}{e}}{ϵ} \stepstoctx \evalctx*{e'}$ and $\TValGpc{e}{τ}$,
  then there is some $e''$, $t'$ such that
  $\evalctx{e}{ϵ} \stepstoctx \evalctx{e''}{t'}$.
\end{lemma}

\begin{proof}
  Assume $\evalctx{\bget{}{e}}{ϵ} \stepstoctx \evalctx*{e'}$ and consider the
  operational semantics in Figures~\ref{fig:full-semantics} and \ref{fig:extensions}.
  We consider a single step in the projected expression which gives us two
  cases.

  If the step happens either entirely within a bracket or entirely outside
  brackets, then the same step must be possible in the original expression.
  If the step in the original expression is not possible because of brackets,
  then there must be brackets around a term other than an arbitrary expression
  or value in the semantic rule employed.
  Since the expression type-checks and by Theorem~\ref{thm:subred}, each
  intermediate expression also type-checks, the protection clause on
  \ruleref{Bracket} and inversion on the protection rules in
  Figure~\ref{fig:core-protect} give us four possible options for the semantic rule
  employed: \ruleref{E-UnPair}, \ruleref{E-BindM}, \ruleref{E-Decl}, and
  \ruleref{E-Endorse}.
  For the first two, we can first step using \ruleref{B-Expand} one or more
  times before applying the original rule.
  For the second two, we can apply \ruleref{B-DeclL}, \ruleref{B-DeclH},
  \ruleref{B-EndorseL}, or \ruleref{B-EndorseH} one or more times again before
  applying the original rule.

  In all cases we see that if the projected term steps and equivalent step is
  possible in the original expression, though possibly after applying one or
  more other steps first.
  This proves the desired result.
\end{proof}

\subsection{Security results}

We now prove the various security results stated throughout the paper.

\begin{lemma}[Release on downgrade]
  \label{lem:down-release}
  Let $\H$ be a high set and $\low = \L \setminus \H$.
  Given a program $e$ such that $\TVal{Γ,x\ty τ_1;\pc}{e}{τ_2}$ with
  $\hightype{\H}{τ_1}$, for all $v_1, v_2$ with $\TValGpc{v_i}{τ_1}$, let
  \[ \evalctx{\subst{e}{x}{\lab*{v_i}}}{v_i} \stepstoctx \evalctx{v'_i}{t_i}. \]
  If $n_1$ and $n_2$ are such that $\nseteq*{\trelt{t^i}{n_i}}{\bullet}$ and
  $\traceeq*{\trpref{t^1}{n_1 - 1}}{\trpref{t^2}{n_2 - 1}}$, then either
  $\trelt{t^i}{n_i} = \downemit*{π}{w_i}$ with $ℓ' ∈ \H$ and $ℓ \notin \H$ for both
  $i = 1,2$, or $\seteq*{\trelt{t^1}{n_1}}{\trelt{t^2}{n_2}}$.
\end{lemma}

\begin{proof}
  We refer to elements of the form $\downemit*{π}{w}$ with $ℓ' ∈ \H$ and $ℓ
  \notin \H$ as _relevant downgrade_ elements.
  This is a proof by induction on the number of relevant downgrade elements in
  $\trpref{t^1}{n_1 - 1}$ and $\trpref{t^2}{n_2 - 1}$.
  Note that while the prefixes can contain any number of relevant downgrade
  elements, the prefixes are $\low$-equivalent, so the downgraded values must
  also be $\low$-equivalent.
  In particular, there must be the same number in each trace prefix.

  As the base case, assume there are no such values in the trace prefixes.
  We first claim that there can be no application of \ruleref{B-DeclL} or
  \ruleref{B-EndorseL} prior to the current step.

  By the typing rules \ruleref{Decl} and \ruleref{Endorse}, that would require
  the value inside the bracket to have type $\says{ℓ'}{τ}$ for some type $τ$.
  The protection clause on \ruleref{Bracket} and the protection rules in
  Figure~\ref{fig:full-protect} would thus require there to be some $H ∈ \H$ such
  that $\stflowjudge*{H}{ℓ'}$, which in turn means $ℓ' ∈ \H$.
  However, application of \ruleref{B-DeclL} or \ruleref{B-EndorseL} requires
  $ℓ \notin \H$.
  Therefore if either rule is applied, the next value emitted to the trace must
  be $\downemit*{π}{w}$ where $ℓ' ∈ \H$ and $ℓ \notin \H$.
  While $\trelt{t^i}{n_i}$ may be of that form, we assumed that no prior values are.

  Next we claim that any differences between the two prefixes must be emitted
  from within brackets.
  The initial expression in each case differs only in the substituted value,
  which is inside brackets.
  If we syntactically examine the expression at every step of evaluation, we
  notice that no semantic rules allow anything inside brackets to affect
  anything outside brackets except through the \ruleref{B-DeclL} and
  \ruleref{B-EndorseL} rules, which remove brackets.
  Since those rules are never applied, all differences must be contained within
  brackets.
  This means that all differences within the prefixes must have been emitted
  from inside brackets.

  Finally, we claim that if a trace element $c$ is emitted from within a
  bracket, then $\seteq*{c}{\bullet}$.
  The only three rules that emit trace elements are \ruleref{E-UnitM},
  \ruleref{E-Decl}, and \ruleref{E-Endorse}.
  By Theorem~\ref{thm:subred}, the expression stepping must type check, and
  each of the corresponding typing rules (\ruleref{UnitM}, \ruleref{Decl}, and
  \ruleref{Endorse}) contain the condition $\stflowjudge*{\pc}{ℓ}$.
  By \ruleref{Bracket}, if the expression is inside a bracket then $\pc ∈ \H$
  and since $\H$ is upward closed, this means $ℓ ∈ \H$.
  Thus by \ruleref{Eq-UnitM} or \ruleref{Eq-Down}, $\seteq*{c}{\bullet}$.

  Coupled with the above logic about applications of \ruleref{B-DeclL} and
  \ruleref{B-EndorseL} and the fact that $\nseteq*{\trelt{t^i}{n_i}}{\bullet}$, either
  $\seteq*{\trelt{t^1}{n_1}}{\trelt{t^2}{n_2}}$, or both are the result of downgrades that
  resulted in applications of either \ruleref{B-DeclL} or \ruleref{B-EndorseL}.
  In the latter case both are of the form $\downemit*{π}{w}$ where $ℓ' ∈ \H$
  and $ℓ \notin \H$, as desired.

  Now we assume that there is at least one such expression in the trace
  prefixes, but those prefixes are still $\low$-equivalent.
  Take the first such trace element.
  This element must appear in both trace prefixes.
  Let $e_1$ and $e_2$ be the top-level expressions that stepped in each trace
  to emit the downgrade element.
  The two expressions can differ inside brackets, and can differ inside the
  downgraded value (that was previously inside brackets).
  Because the downgraded values are $\low$-equivalent, any such differences
  must be contained inside terms of the form $\vreturn{ℓ}{w}$ where $ℓ ∈ \H$.
  We can replace any such terms by $\lab*{\vreturnp{ℓ}{w}}$ in the expression
  and it will still type-check since $ℓ ∈ \H$ and $w$ type-checks in any $\pc$.
  This results in a new pair of expressions that are $\low$-equivalent except
  inside brackets.
  By Lemmas~\ref{lem:sound} and \ref{lem:complete} these expression generate
  traces that are equivalent to the original up to brackets and bullets.
  By construction, these new expressions generate trace prefixes with one less
  $\downemit*{π}{w}$ element than the original expressions, so by the inductive
  hypothesis, the result holds.

  Since the same brackets were necessarily added in both expressions, the new
  expressions will thus generate equivalent traces if and only if the old
  expressions did.
  Since the trace prefixes prior to this modification were equivalent, if the
  new expressions generate equivalent traces, the original expressions must
  also generate equivalent traces.
  Thus the result holds for the original traces as well.
\end{proof}

\begin{retheorem}{thm:no-downgrade-ni}[Noninterference of non-downgrading programs]
  Let $\H$ be a high set and let $\low = \L \setminus \H$.
  Given an expression $e$ such that $\TVal{Γ,x\ty τ_1;\pc}{e}{τ_2}$ where
  $\hightype{\H}{τ_1}$, for all $v_1, v_2$ with $\TValGpc{v_i}{τ_1}$, if
  \[ \evalctx{\subst{e}{x}{v_i}}{v_i} \stepstoctx \evalctx{v'_i}{t^i} \]
  then either there is some $\downemit*{π}{w} ∈ t^i$ where $ℓ' ∈ \H$ and
  $ℓ \notin \H$, or $\traceeq*{t^1}{t^2}$.
\end{retheorem}

\begin{proof}
  Since $\TValGpc{v_i}{τ_1}$ and $\hightype{\H}{τ_1}$, \ruleref{Bracket} says
  $\TValGpc{\lab*{v_i}}{τ_1}$.
  Lemmas~\ref{lem:sound} and \ref{lem:complete} tell us that the result holds
  as stated above if and only if it holds when substituting $\lab*{v_i}$
  instead of $v_i$.
  From there we can apply Lemma~\ref{lem:down-release}.
  If there are no $\downemit*{π}{w}$ events in either trace with $ℓ' ∈ \H$ and
  $ℓ \notin \H$, then each $\low$-visible term in each trace must be equivalent
  to each other by induction on the number of elements (empty traces are
  equivalent).
\end{proof}

\begin{retheorem}{thm:high-pc-ni}[Noninterference of high-$\pc$ programs]
  Let $\A$ be an attacker inducing high sets $\U$ and $\secret$.
  Let $\H$ be one of those high sets and $\low = \L \setminus \H$.
  Given an expression $e$ such that $\TVal{Γ,x\ty τ_1;\pc}{e}{τ_2}$ where
  $\hightype{\H}{τ_1}$, for all $v_1, v_2$ with $\TValGpc{v_i}{τ_1}$, if
  $\evalctx{\subst{e}{x}{v_i}}{v_i} \stepstoctx \evalctx{v'_i}{t^i}$
  and $\pc ∈ \U ∪ \secret$, then $\traceeq*{t^1}{t^2}$.
\end{retheorem}

\begin{proof}
  We claim that neither trace contains any $\downemit*{π}{w}$ elements where
  $ℓ' ∈ \H$ and $ℓ \notin \H$, thus reducing this to
  Theorem~\ref{thm:no-downgrade-ni}.
  There are three cases to consider here: $\pc ∈ \H$, $\pc ∈ \U$ and
  $\H = \secret$, and $\pc ∈ \secret$ and $\H = \U$.

  \begin{case}[$\pc ∈ \H$]
    Both \ruleref{Decl} and \ruleref{Endorse} require $\stflowjudge*{\pc}{ℓ}$,
    so by upward-closure of $\H$, $ℓ ∈ \H$.
  \end{case}

  \begin{case}[$\pc ∈ \U$ and $\H = \secret$]
    \ruleref{Decl} contains the condition
    $\stflowjudge*{ℓ'^{→}}{ℓ^{→} ⊔ \view{(ℓ' ⊔ \pc)^{←}}}$.
    Converting into the authority lattice, this gives us
    $\stafjudge*{ℓ^{→} ∧ \view{(ℓ' ∨ \pc)^{←}}}{ℓ'^{→}}$.
    Since $\U = \A^{←}$, Condition~\ref{li:atk:disj-with-a} of
    Definition~\ref{def:attacker} means $(ℓ' ∨ \pc)^{←} ∈ \A^{←}$ and thus
    Condition~\ref{li:atk:voice-and-view} means that $\view{(ℓ' ∨ \pc)^{←}} ∈ \A^{→}$,
    so by Condition~\ref{li:atk:a-closed-conj} of Definition~\ref{def:attacker}, if
    $ℓ^{→} ∈ \A^{→}$, then by transitivity $ℓ'^{→} ∈ \A^{→}$.
    Since $\H = \secret = \L \setminus (\A^{→})$, this means $ℓ \notin \H$ only
    if $ℓ' \notin \H$.
  \end{case}

  \begin{case}[$\pc ∈ \secret$ and $\H = \U$]
    This argument is nearly identical to the previous case.
    \ruleref{Endorse} means $\stflowjudge*{ℓ'^{←}}{ℓ^{←} ⊔ \voice{(ℓ' ⊔ \pc)^{→}}}$
    which, in the authority lattice, means
    $\stafjudge*{ℓ'^{←}}{ℓ^{←} ∨ \voice{(ℓ' ∧ \pc)^{→}}}$.
    Since $\pc ∈ \secret = \L \setminus (\A^{→})$,
    Condition~\ref{li:atk:conj-with-not-a} of Definition~\ref{def:attacker}
    ensures $(ℓ' ∧ \pc)^{→} ∈ \secret$.
    Thus Condition~\ref{li:atk:voice-and-view} means $\voice{(ℓ' ∧ \pc)^{→}} \notin \A^{←} = \H$,
    so Condition~\ref{li:atk:not-a-closed-disj} and transitivity mean that
    $ℓ^{←} \notin \H$ only if $ℓ'^{←} \notin \H$.
  \end{case}
\end{proof}

\begin{retheorem}{thm:secret-untrust-ni}[Noninterference of secret–untrusted data]
  Let $\A$ be an attacker inducing high sets $\U$ and $\secret$.
  Let $\H = \U ∩ \secret$ and $\low = \L \setminus \H$.
  Given an expression $e$ such that $\TVal{Γ,x\ty τ_1;\pc}{e}{τ_2}$ where
  $\hightype{\H}{τ_1}$, for all $v_1, v_2$ with $\TValGpc{v_i}{τ_1}$, if
  $\evalctx{\subst{e}{x}{v_i}}{v_i} \stepstoctx \evalctx{v'_i}{t^i}$
  then $\traceeq*{t^1}{t^2}$.
\end{retheorem}

\begin{proof}
  First we note that $\H$ is a high set as the intersection of two upward
  closed sets is also upward closed.
  We claim that neither trace contains any $\downemit*{π}{w}$ elements where
  $ℓ' ∈ \H$ and $ℓ \notin \H$, thus reducing this to
  Theorem~\ref{thm:no-downgrade-ni}.
  We will cover the two cases of $π$ separately despite their similarities.

  \begin{case}[$π = {→}$]
    The only way to emit $\downemit*{→}{w}$ is through \ruleref{E-Decl}.
    By Theorem~\ref{thm:subred}, each intermediate expression type-checks under
    the same initial context.
    This is a sub-expression, meaning there is some $Γ'$, $\pc'$, and $τ$ such that
    \[ \TVal{Γ';\pc'}{\cdowngrade{\vreturn{ℓ'}{v}}{ℓ}}{\says{ℓ}{τ}}. \]
    Therefore \ruleref{Decl} tells us that
    $\stflowjudge*{ℓ'^{→}}{ℓ^{→} ⊔ \view{(ℓ' ⊔ \pc')^{←}}}$.
    Every typing rule either moves $\pc$ up the lattice (by joining it with
    another label) or leaves it unchanged, so $\stflowjudge*{\pc}{\pc'}$ and
    therefore $\stflowjudge*{\view{\pc'^{←}}}{\view{\pc^{←}}}$, meaning
    $\stflowjudge*{ℓ'^{→}}{ℓ^{→} ⊔ \view{(ℓ' ⊔ \pc)^{←}}}$.
    Converting to the authority lattice gives us
    $\stafjudge*{ℓ^{→} ∧ \view{(ℓ' ∨ \pc)^{←}}}{ℓ'^{→}}$.

    If $ℓ' ∈ \H$, then $ℓ'^{←} ∈ \U = \A^{←}$, so by the same argument as in
    the proof of Theorem~\ref{thm:high-pc-ni}, Definition~\ref{def:attacker}
    gives us that $\view{(ℓ' ∨ \pc)^{←}} ∈ \A^{→}$.
    Thus by Condition~\ref{li:atk:a-closed-conj} of
    Definition~\ref{def:attacker} and transitivity, if $ℓ^{→} ∈ \A^{→}$ it must
    also be the case that $ℓ'^{→} ∈ \A^{→}$.
    However, $ℓ' ∈ \H$, so $ℓ'^{→} ∈ \secret = \L \setminus (\A^{→})$.
    Therefore $ℓ^{→} ∈ \secret$.
    Since $ℓ' ∈ \H$ and $ℓ'^{←} = ℓ^{←}$, this means $ℓ ∈ \H$.
  \end{case}

  \begin{case}[$π = {←}$]
    Dual to the above case, $\downemit*{←}{w}$ must be emitted from
    \ruleref{E-Endorse}, so Theorem~\ref{thm:subred} now implies that there is
    some $Γ'$, $\pc'$, and $τ$ such that
    \[ \TVal{Γ';\pc'}{\idowngrade{\vreturn{ℓ'}{v}}{ℓ}}{τ}. \]
    Since $\stflowjudge*{\pc}{\pc'}$ by the same logic as above,
    \ruleref{Endorse} tells us
    $\stflowjudge*{ℓ'^{←}}{ℓ^{←} ⊔ \voice{(ℓ' ⊔ \pc)^{→}}}$.
    If $ℓ' ∈ \H$, then $ℓ^{←} ⊔ \voice{(ℓ' ⊔ \pc)^{→}} ∈ \H$.
    Since $ℓ'^{→} ∈ \H^{→} = \L \setminus (\A^{→})$, we know that
    $\voice{(ℓ' ⊔ \pc)^{→}} \notin \A^{←}$ by the same argument as in the proof
    of Theorem~\ref{thm:high-pc-ni}.
    Since $ℓ^{←} ⊔ \voice{(ℓ' ⊔ \pc)^{→}} = ℓ^{←} ∨ \voice{(ℓ' ∧ \pc)^{→}}$,
    Condition~\ref{li:atk:not-a-closed-disj} of Definition~\ref{def:attacker}
    requires that $ℓ^{←} ∈ \A^{←} = \U$.
    Since $ℓ^{→} = ℓ'^{→} ∈ \secret$, we see that $ℓ ∈ \H$.
  \end{case}

  Thus we see that for any trace element $\downemit*{π}{w}$, if $ℓ' ∈ \H$, then
  $ℓ ∈ \H$, so by Theorem~\ref{thm:no-downgrade-ni} the result holds.
\end{proof}

\begin{retheorem}{thm:nmifc}[Nonmalleable information flow]
  For any program $e$ such that $\TVal{Γ,x\ty τ_x,y\ty τ_y;\pc}{e}{τ}$, $e$ enforces NMIF.
\end{retheorem}

\begin{proof}
  We provide here only a proof of case~\ref{li:nmifc:rd} of Definition~\ref{def:nmifc}.
  All statements in case~\ref{li:nmifc:te} are exactly dual so a precisely dual
  argument holds.

  Let $\A$ be an attacker inducing high sets $\U$ and $\secret$ and let
  $\trusted = \L \setminus \U$, $\public = \L \setminus \secret$, and
  $\low = \trusted ∩ \public$.
  To prove case~\ref{li:nmifc:rd} we prove a contrapositive of the stated
  implication.
  Specifically, if $\ntraceeq{\public}{\trpref{t^{12}}{n_{12}}}{\trpref{t^{22}}{n_{22}}}$ but
  $\rel*{←}{w_1}$, then $\ntraceeq{\public}{\trpref{t^{11}}{n_{11}}}{\trpref{t^{21}}{n_{21}}}$.

  First we note that the theorem is uninteresting unless
  $\hightype{\secret}{τ_x}$ and $\hightype{\U}{τ_y}$.
  Since $v_i$ and $w_j$ are both present in trace $t^{ij}$, if
  $\nseteq{\public}{v_1}{v_2}$ or $\nseteq{\trusted}{w_1}{w_2}$, the theorem is
  trivially true.
  In the former case, varying the first input clearly results in non-equivalent
  traces for both attacks.
  In the latter case, the result holds for $n_{ij} = 1$ and the precondition is
  clearly false otherwise.
  Now consider the case where $\seteq{\public}{v_1}{v_2}$ and
  $\seteq{\trusted}{w_1}{w_2}$.
  If $\hightype{\secret}{τ_x}$ and $\hightype{\U}{τ_y}$, then these are
  trivially true---the interesting case we will handle below.
  Otherwise the equivalences require the contents of the value to be identical
  (except for nested secret/untrusted values).
  Any identical values clearly cannot result in distinguishable traces.
  We can handle nested secret/untrusted values by fixing the public/trusted
  parts of the inputs and viewing those values as the inputs instead, thus
  reducing to the case where $\hightype{\secret}{τ_x}$ and
  $\hightype{\U}{τ_y}$.\footnote{This technically reduces to the many-input
  version of the theorem that we prove is equivalent in
  Appendix~\ref{sec:nmif-generalize}.}
  The rest of the proof assumes we are in this case.

  By assumption $\ntraceeq{\public}{\trpref{t^{12}}{n_{12}}}{\trpref{t^{22}}{n_{22}}}$,
  there is some point at which the traces become distinguishable.
  This can happen for one of two reasons: one trace is a prefix of the other
  (up to extra $\bullet$-equivalent terms), or there are two non-$\bullet$
  terms that are non-equivalent.

  In the first case, without loss of generality assume $\trpref{t^{12}}{n_{12}}$ is a
  prefix of $\trpref{t^{22}}{n_{22}}$.
  We claim that  $\trpref{t^{11}}{n_{11}}$ is a prefix of $\trpref{t^{21}}{n_{21}}$.
  Lemma~\ref{lem:down-release} ensures that any difference must come from a
  downgrade, which must be a declassification because \ruleref{Endorse}
  prohibits changing the confidentiality and $\public$ allows labels of any
  integrity.
  By the argument in Theorem~\ref{thm:secret-untrust-ni} this declassification
  event must be on trusted data, thus producing a public–trusted output.
  Thus the condition that
  $\traceeq{\trusted}{\trpref{t^{i1}}{n_{i1} - 1}}{\trpref{t^{i2}}{n_{i2} - 1}}$ ensures that
  any declassifications appearing in $\trpref{t^{21}}{n_{21} - 1}$ appear identically
  in $\trpref{t^{11}}{n_{11} - 1}$, and similarly for $\trpref{t^{22}}{n_{22} - 1}$ and
  $\trpref{t^{12}}{n_{12} - 1}$.
  Since $\nseteq*{\trelt{t^{ij}}{n_{ij}}}{\bullet}$, any public–untrusted events
  appearing in $\trpref{t^{2j}}{n_{2j}}$ must also appear in $\trpref{t^{1j}}{n_{1j}}$, and
  thus $\trpref{t^{11}}{n_{11}}$ must be a prefix of $\trpref{t^{21}}{n_{21}}$.

  For the second case---two non-$\bullet$ terms are non-equivalent---let
  $n'_{i2}$ be the indices of the first such terms in each trace.
  That is, $\nseteq{\public}{\trelt{t^{i2}}{n'_{i2}}}{\bullet}$ and
  $\nseteq{\public}{\trelt{t^{12}}{n'_{12}}}{\trelt{t^{22}}{n'_{22}}}$, but
  $\traceeq{\public}{\trpref{t^{12}}{n'_{12} - 1}}{\trpref{t^{22}}{n'_{22} - 1}}$.
  Again Lemma~\ref{lem:down-release} ensures that the first non-equivalence must
  be the result of a declassification that exists in both traces.

  By construction $n'_{ij} ≤ n_{ij}$ and we have that
  $\traceeq{\public}{\trpref{t^{12}}{n'_{12}-1}}{\trpref{t^{22}}{n'_{22}-1}}$ by the definition
  of $n'_{ij}$ and $\traceeq{\trusted}{\trpref{t^{i1}}{n_{i1}-1}}{\trpref{t^{i2}}{n_{i2}-1}}$
  by assumption.
  Consequently, for all $i, j, k, l ∈ \{1,2\}$, we have
  \[ \traceeq*{\trpref{t^{ij}}{n'_{ij}-1}}{\trpref{t^{kl}}{n'_{kl}-1}}. \]
  In particular, this is true for $t^{11}$ and $t^{21}$.
  We also know that for all $i,j ∈ \{1,2\}$ $\trelt{t^{ij}}{n'_{ij}} =
  \downemit*{→}{w_{ij}}$ for some value $w_{ij}$.

  We know that $\nseteq{\public}{w_{12}}{w_{22}}$, and we now claim that
  $\nseteq{\public}{w_{11}}{w_{21}}$.
  To emit this value, \ruleref{E-Decl} must have been applied, meaning there
  must have been a term $\vreturn{ℓ'}{w_{ij}}$ in the preceding expression.
  Since such expressions are not in the source language---they can only be
  created by applications of \ruleref{E-UnitM}, \ruleref{E-Decl}, and
  \ruleref{E-Endorse}---this expression must appear previously in each trace.
  We know that $ℓ ∈ \public ∩ \trusted$ and \ruleref{Decl} requires that
  $ℓ'^{←} = ℓ^{←}$, so therefore $ℓ' ∈ \trusted$.
  Since the prefixes prior to this event are trusted-equivalent when varying
  only attacks, $\seteq{\trusted}{w_{i1}}{w_{i2}}$ for $i = 1,2$.
  By assumption $\nseteq{\public}{w_{12}}{w_{22}}$ and this is the first
  difference in those traces.
  This means that if $\seteq{\public}{w_{11}}{w_{21}}$ then $w_{12}$ and
  $w_{22}$ must have differed only in untrusted public values nested inside the
  original declassification (i.e., $\weakseteq{\A}{→}{w_{12}}{w_{22}}$).
  Therefore if $\seteq{\public}{w_{11}}{w_{21}}$, then $v_1$, $v_2$, and $w_2$
  are exactly the inputs needed in Definition~\ref{def:irrel-input} to
  demonstrate that $w_1$ is an irrelevant attack.
  We assumed that $\rel*{←}{w_1}$, so $\nseteq{\public}{w_{11}}{w_{21}}$ and
  therefore $\ntraceeq{\public}{\trpref{t^{11}}{n_{11}}}{\trpref{t^{21}}{n_{21}}}$, proving our
  desired result.
\end{proof}

\end{document}